\RequirePackage{easy,etex} 
\documentclass{llncs}
\usepackage{llncsdoc}
\usepackage{graphicx}
\usepackage{amssymb}
\usepackage{amsmath}
\usepackage{hyperref}
\usepackage{todonotes}
\usepackage{enumitem}
\usepackage{verbatim} 
\usepackage{gensymb} 
\usepackage{tikz}
\usepackage{color}
\usepackage{plaatjes}
\usepackage{array}
\usepackage[capitalise]{cleveref}
\usepackage{enumitem}
\setlist{nolistsep}

\graphicspath{{.}}

\newcommand{\BT}[1]{\includegraphics[width=1.65em]{B#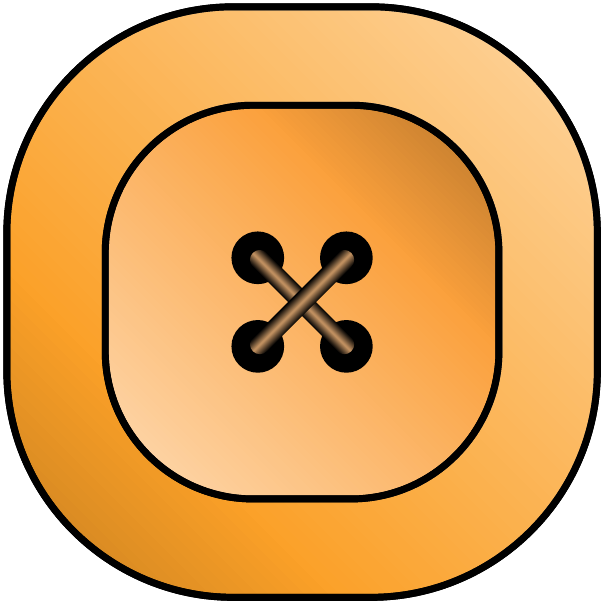}}

\def\comic#1#2#3{\parbox{#1}{\centering\includegraphics[width=#1]{#2}\\{\footnotesize #3}}}
\def\comicII#1#2{\parbox{#1}{\centering\includegraphics[width=#1]{#2}}}

\newcommand{\remarkx}[2]{\todo{\textcolor{blue}{\textsc{#1:}} \textcolor{red}{\textsf{#2}}}}
\newcommand{\aaronw}[1]{\remarkx{AaronW}{#1}}

\newenvironment{sketch}{\paragraph{Proof Sketch:}}{\hfill$\square$}

\newcommand{\cclass}[1]{\ensuremath{\mathord{\rm #1}}}
\newcommand{\ruleset}[1]{\textsc{#1}}
\newcommand{\impartial}[0]{\ruleset{Impartial}}
\newcommand{\scoring}[0]{\ruleset{Scoring}}

\DeclareRobustCommand{\dirsFour}{%
\begin{tikzpicture}%
\draw (-0.2ex,-0.2ex);%
\draw (1.2ex,1.2ex);%
\draw (0,0) -- (1ex,1ex);%
\draw (-0.2ex,0.5ex) -- (1.2ex,0.5ex);%
\draw (1ex,0) -- (0,1ex);%
\draw (0.5ex,-0.2ex) -- (0.5ex,1.2ex);%
\end{tikzpicture}%
}
\DeclareRobustCommand{\dirsThree}{%
\begin{tikzpicture}%
\draw (-0.2ex,-0.2ex);%
\draw (1.2ex,1.2ex);%
\draw (-0.2ex,0.5ex) -- (1.2ex,0.5ex);%
\draw (0.15ex,-0.1ex) -- (0.85ex,1.1ex);%
\draw (0.15ex,1.1ex) -- (0.85ex,-0.1ex);%
\end{tikzpicture}%
}
\DeclareRobustCommand{\dirsTwoA}{%
\begin{tikzpicture}%
\draw (-0.2ex,-0.2ex);%
\draw (1.2ex,1.2ex);%
\draw (-0.2ex,0.5ex) -- (1.2ex,0.5ex);%
\draw (0.5ex,-0.2ex) -- (0.5ex,1.2ex);%
\end{tikzpicture}%
}
\DeclareRobustCommand{\dirsTwoB}{%
\begin{tikzpicture}%
\draw (-0.2ex,-0.2ex);%
\draw (1.2ex,1.2ex);%
\draw (0,0) -- (1ex,1ex);%
\draw (-0.2ex,0.5ex) -- (1.2ex,0.5ex);%
\end{tikzpicture}%
}
\DeclareRobustCommand{\dirsOne}{%
\begin{tikzpicture}%
\draw (-0.2ex,-0.2ex);%
\draw (1.2ex,1.2ex);%
\draw (-0.2ex,0.5ex) -- (1.2ex,0.5ex);%
\end{tikzpicture}%
}

\newcommand{\CButt}[2][]{\circleButton[#2]{#1}}
\newcommand{\VCButt}[2][]{\squareButton[#2]{#1}}
\newcommand{\circleButton}[2][]{\tikz[baseline=(char.base)]{
            \node[shape=circle,draw,inner sep=0.7pt] (char) {#1};}_{\!\raisebox{2pt}{$\scriptstyle {#2}$}}}
\newcommand{\squareButton}[2][]{\tikz[baseline=(char.base)]{
            \node[shape=rectangle,draw,inner sep=1.2pt] (char) {#1};}_{\!\raisebox{0pt}{$\scriptstyle {#2}$}}}

\newcommand{\Board}[6]{B\&S[#1,#2,#3,#4,#5](#6)}

\newcounter{section-⁠preserve}
\newcounter{theorem-⁠preserve}
\newcommand{\blank}[1]{}
\newtoks\magicAppendix
\magicAppendix={}
\newtoks\magictoks
\newif\iflater
\laterfalse
\long\def\later#1{\magictoks={#1}%
  \edef\magictodo{\noexpand\magicAppendix={\the\magicAppendix \par
    \the\magictoks%
  }}
  \magictodo}
\long\def\both#1{\magictoks={#1}%
  \edef\magictodo{\noexpand\magicAppendix={\the\magicAppendix \par
    \noexpand\setcounter{theorem-⁠preserve}{\noexpand\arabic{theorem}}%
    \noexpand\setcounter{theorem}{\arabic{theorem}}%
    \noexpand\setcounter{section-⁠preserve}{\noexpand\arabic{section}}%
    \noexpand\setcounter{section}{\arabic{section}}%
    \noexpand\let\noexpand\oldsection=\noexpand\thesection
    \noexpand\def\noexpand\thesection{\thesection}
    \noexpand\let\noexpand\oldlabel=\noexpand\label
    \noexpand\let\noexpand\label=\noexpand\blank
    \the\magictoks%
    \noexpand\setcounter{theorem}{\noexpand\arabic{theorem-⁠preserve}}%
    \noexpand\setcounter{section}{\noexpand\arabic{section-⁠preserve}}%
    \noexpand\let\noexpand\thesection=\noexpand\oldsection
    \noexpand\let\noexpand\label=\noexpand\oldlabel
  }}
  \magictodo
  \the\magictoks}
\def\magicappendix{\latertrue \the\magicAppendix}

\newif\ifabstract
\abstractfalse
\newif\iffull
\ifabstract \fullfalse \else \fulltrue \fi

\title{Single-Player and Two-Player Buttons \& Scissors Games}

\begin{document}

\makeatletter
\let\@fnsymbol=\@arabic
\makeatother

\author{
Kyle Burke\thanks{
Plymouth State University, {\tt kgburke@plymouth.edu}}
\and
Erik D. Demaine\thanks{
Massachusetts Institute of Technology, \texttt{\{edemaine,achester\}@mit.edu}}
\and
Harrison Gregg\thanks{
Bard College at Simon's Rock,\protect\\ {\tt \{hgregg11,jleonard11,asantiago11,awilliams\}@simons-rock.edu}}
\and
Robert A. Hearn\thanks{{\tt bob@hearn.to}}
\and
Adam Hesterberg$^2$
\and
Michael Hoffmann\thanks{
ETH Z{\"u}rich, {\tt hoffmann@inf.ethz.ch}}
\and
Hiro Ito\thanks{
The University of Electro-Communications, {\tt itohiro@uec.ac.jp}}
\and
Irina Kostitsyna\thanks{
Technische Universiteit Eindhoven, {\tt i.kostitsyna@tue.nl}. 
Supported in part by
    NWO project no. 639.023.208.}
\and
Jody Leonard$^3$
\and
Maarten L{\"o}ffler\thanks{
Universiteit Utrecht, 
 {\tt m.loffler@uu.nl}}
\and
Aaron Santiago$^3$
\and
Christiane Schmidt\thanks{
Link\"oping University, {\tt christiane.schmidt@liu.se}. Supported in part by grant 2014-03476 from Sweden's innovation agency VINNOVA.}
\and
Ryuhei Uehara\thanks{
Japan Advanced Institute of Science and Technology, {\tt uehara@jaist.ac.jp}}
\and
Yushi Uno\thanks{
Osaka Prefecture University, {\tt uno@mi.s.osakafu-u.ac.jp}}
\and
Aaron Williams$^3$
        }

\institute{} 

\maketitle


\begin{abstract}
  We study the computational complexity of the Buttons \& Scissors
  game and obtain sharp thresholds with respect to several
  parameters. Specifically we show that the game is NP-complete for
  $C=2$ colors but polytime solvable for $C=1$. Similarly the game is
  NP-complete if every color is used by at most $F=4$ buttons but
  polytime solvable for $F\le 3$. We also consider restrictions on the
  board size, cut directions, and cut sizes. Finally, we introduce
  several natural two-player versions of the game and show that they
  are PSPACE-complete.
\end{abstract}



\section{Introduction}\label{sec:intro}

Buttons $\&$ Scissors is a single-player puzzle by KyWorks. 
The goal of each level is to remove every button by a sequence of horizontal, vertical, and diagonal cuts, as illustrated by \cref{fig:LevelSeven}.
It is NP-complete to decide if a given level is solvable~\cite{glsw-bsnpc-15}. 
We study several restricted versions of the game and show that some remain hard, whereas others can be solved in polynomial time. 
We also consider natural extensions to two player games which turn out to be PSPACE-complete.

\begin{figure}
\centering
\begin{tabular}{ccccc}
\raisebox{1.25em}{\includegraphics[height=7.5em]{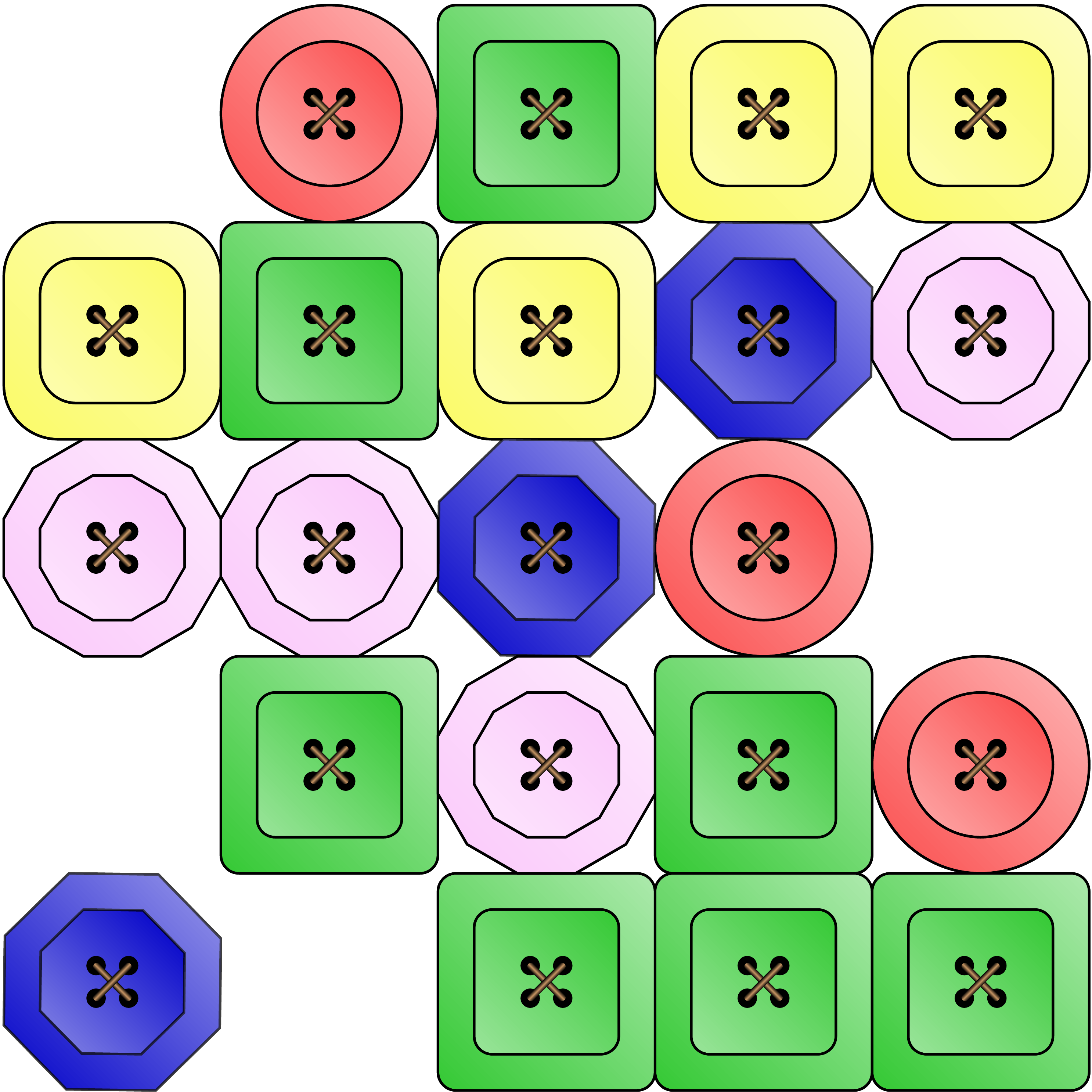}} 
  & \hspace{3em} &
\raisebox{1.25em}{\includegraphics[height=7.5em]{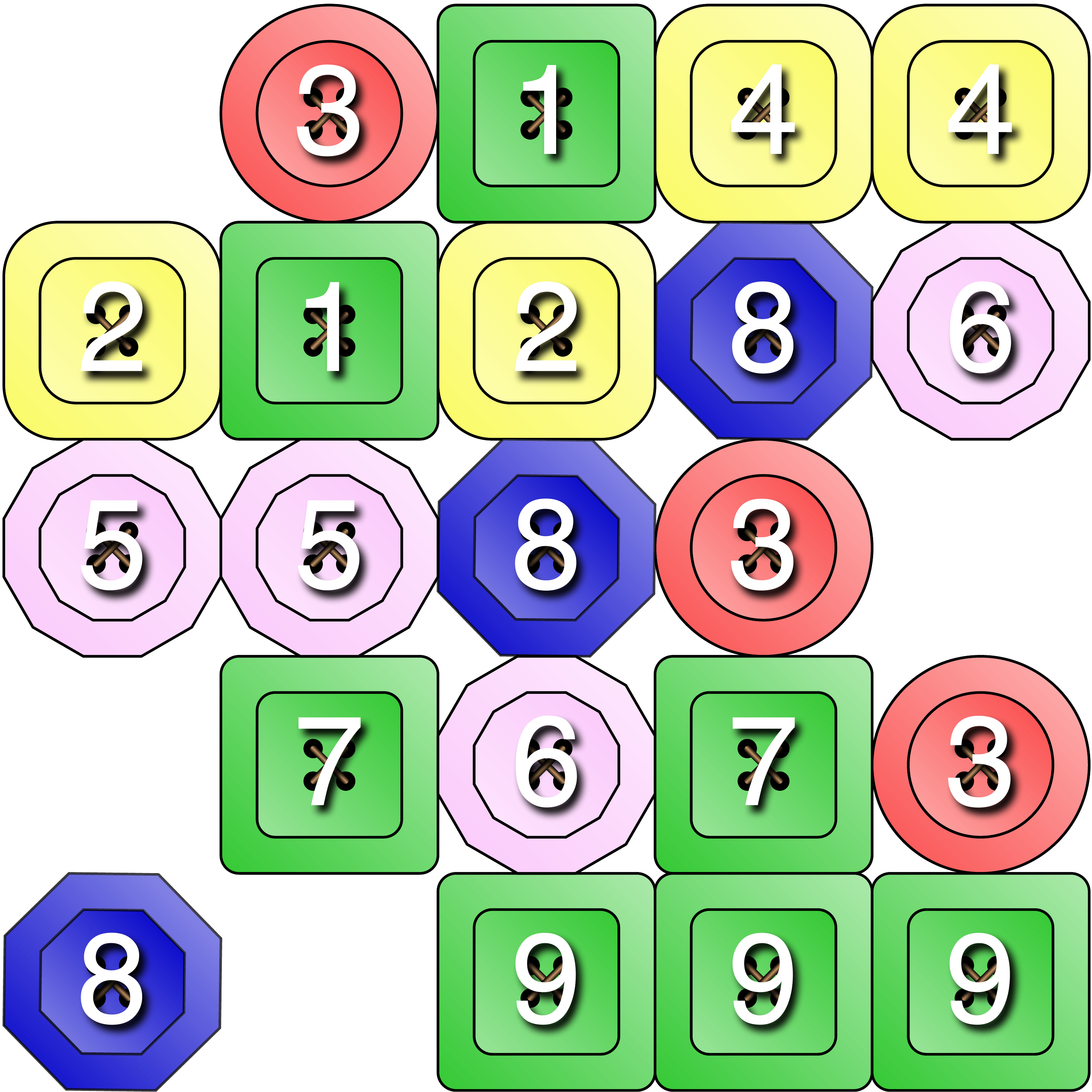}} 
  & \hspace{3em} &
\includegraphics[height=10em]{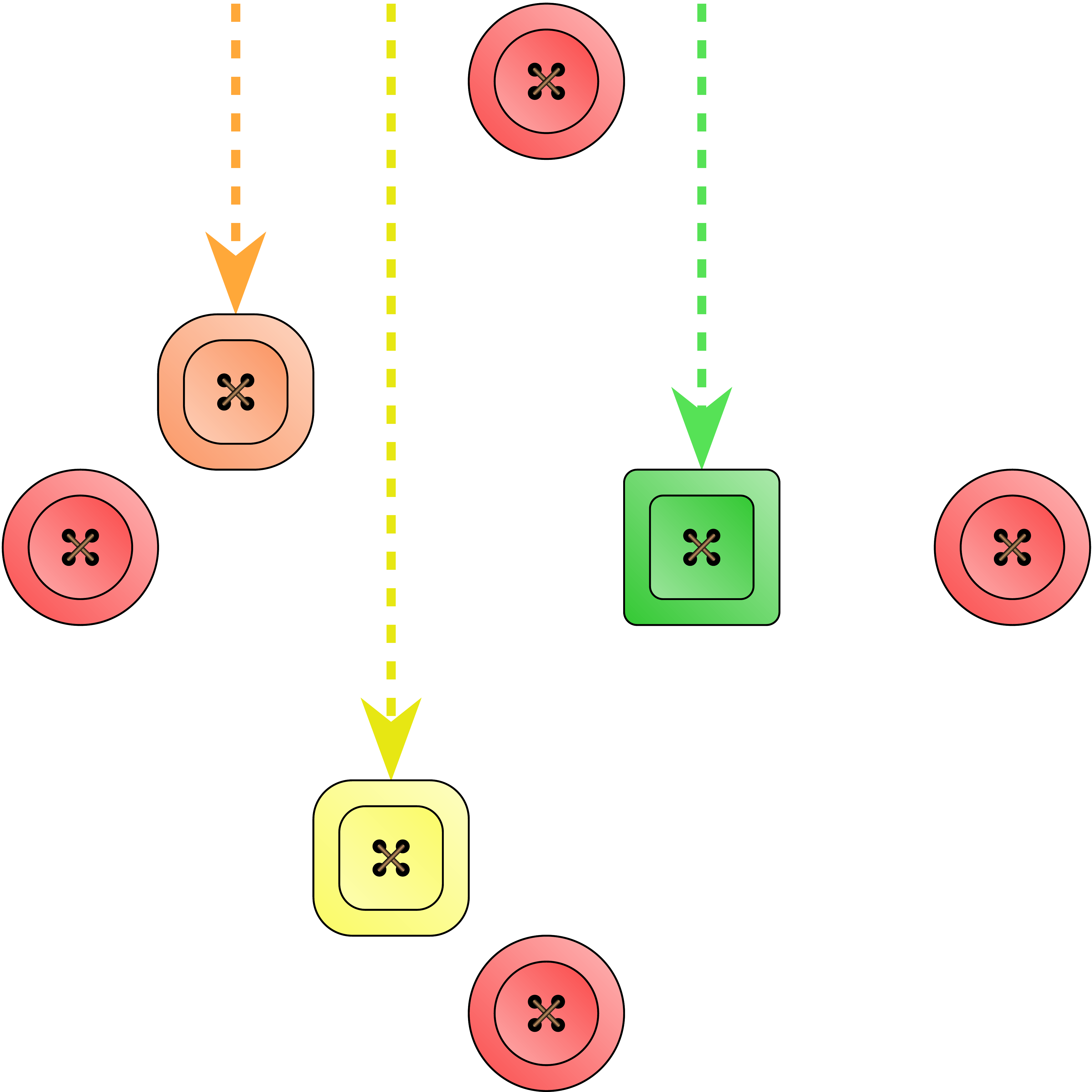} \\  
(a) && (b) && (c)
\end{tabular}
\caption{(a) Level 7 in the Buttons \& Scissors app is an $m \times n = 5 \times 5$ grid with $C=5$ colors, each used at most $F=7$ times; (b) a solution using nine cuts with sizes in $S = \{2,3\}$ and directions $d = \dirsThree$ (no vertical cut is used); (c) a gadget used in \cref{thm:freq4}.}
\label{fig:LevelSeven}
\end{figure}

\cref{sec:prelim} begins with preliminaries, then we discuss one-player puzzles in \cref{sec:1p} and two-player games in \cref{sec:2p}.
Open problems appear in \cref{sec:open}.
\ifabstract
Due to space restrictions, some proofs are sketched or appear in the appendix.
\fi

\section{Preliminaries}\label{sec:prelim}


A Buttons $\&$ Scissors {\it board} $B$ is an $m\times n$ grid, where each grid position is either empty or occupied by a button with one of $C$ different colors. 
A {\it cut} is given by two distinct buttons $b_1, b_2$ of the same color $c$ that share either the $x$-coordinate, the $y$-coordinate, or are located on the same diagonal ($45^{\degree}$ and ${-}45^{\degree}$). 
The \emph{size} $s$ of a cut is the number of buttons on the line segment $\overline{b_1 b_2}$ and so $s \geq 2$.
A cut is {\it feasible} for $B$ if $\overline{b_1 b_2}$ only contains buttons of a single color.

When a feasible cut is applied to a board $B$, the resulting board $B'$ is obtained by substituting the buttons of color $c$ on $\overline{b_1 b_2}$ with empty grid entries. A {\it solution} to board $B$ is a sequence of boards and feasible cuts $B_1, x_1, B_2, x_2, \ldots, B_t, x_t, B_{t+1}$, where $B_{t+1}$ is empty, and each cut $x_i$ is feasible for $B_i$ and creates $B_{i+1}$.


Each instance can be parameterized as follows (see \cref{fig:LevelSeven} for an example):
\begin{enumerate}[noitemsep,nolistsep]
\item The {\it board size} $m \times n$.
\item The {\it number of colors} $C$.
\item The \emph{maximum frequency} $F$ of an individual color.
\item The {\it cut directions} $d$ can be limited to $d \in \{\dirsFour, \dirsThree, \dirsTwoA, \dirsOne\}$.
\item The {\it cut size set} $S$ limits feasible cuts to having size $s \in S$.
\end{enumerate}
Each $d \in \{\dirsFour, \dirsThree, \dirsTwoA, \dirsOne\}$ is a set of cut directions (i.e. $\dirsTwoA$ for horizontal and vertical). 
We limit ourselves to these options because an $m\times n$ board can be rotated $90^{\circ}$ to an equivalent $n\times m$ board, or $45^{\circ}$ to an equivalent $k\times k$ board for $k=m+n-1$ with blank squares.
Similarly, we can shear the grid by padding row $i$ with $i-1$ blanks on the left and $m-i$ blanks on the right which converts $d = \dirsTwoA$ to $d = \dirsTwoB$.
%
We obtain the family of games below ($\Board{n \times n}{\infty}{\infty}{\dirsFour}{$\{2,3\}$}{B}$ is the~original):\smallskip

\noindent{\bf Decision Problem:}  $\Board{m \times n}{C}{F}{d}{S}{B}$. \\
{\bf Input:} An $m \times n$ board $B$ with buttons of $C$ colors, each used at most $F$ times.\\
{\bf Output:} True $\iff$ $B$ is solvable with cuts of size $s \in S$ and directions $d$.\smallskip



Now we provide three observations for later use.
First note that a single cut of size $s$ can be accomplished by cuts of size $s_1, s_2, \ldots, s_k$ so long as $s = s_1 + s_2 + \cdots + s_k$ and $s_i \geq 2$ for all $i$.
Second note that removing all buttons of a single color from a solvable instance cannot result in an unsolvable instance. 

\begin{remark} \label{rem:2or3cuts}
A board can be solved with cut sizes $S = \{2,3,\ldots\}$ if and only if it can be solved with cut sizes $S' = \{2,3\}$.
Also, $\{3,4,\ldots\}$ and $\{3,4,5\}$ are equivalent.
\end{remark}


\begin{remark} \label{rem:removeColor}
If board $B'$ is obtained from board $B$ by removing every button of a single color, then $\Board{m \times n}{C}{F}{d}{S}{B} \implies \Board{m \times n}{C}{F}{d}{S}{B'}$. 
\end{remark}



\section{Single-Player Puzzle}\label{sec:1p}



\subsection{Board Size} \label{sec:1p_size}

We solve one row problems below, and give a conjecture for two rows in \cref{sec:open}.

\begin{theorem} \label{thm:1row}
$\Board{1 \times n}{\infty}{\infty}{\dirsOne}{\{2,3\}}{B}$ is polytime solvable.
\end{theorem}
\begin{proof}
Consider the following context-free grammar,
\begin{equation*}
 S \rightarrow \varepsilon \; \vert \; \square \; \vert \; SS \; \vert \; xSx \; \vert \; xSxSx
\end{equation*}
where $\square$ is an empty square and $x \in \{1,2,\ldots,C\}$.
By \cref{rem:2or3cuts}, the solvable $1\times n$ boards are in one-to-one correspondence with the strings in this language. \qed
\end{proof}

\subsection{Number of Colors} \label{sec:1p_colors}


\subsubsection {Hardness for $2$ colors.}

We begin with a straightforward reduction from 3SAT.
The result will be strengthened later by Theorem \ref{thm:2colors2cuts} using a more difficult proof.

\begin{theorem} \label{thm:2colorsSAT}
  $\Board{n \times n}{2}{\infty}{\dirsTwoA}{\{2,3\}}{B}$ is NP-complete.
\end{theorem}

\begin{sketch}
A \emph{variable gadget} has its own row with exactly three buttons. The middle button is alone in its column, and must be matched with at least one of the other two in the variable row. If the left button is not used in this match, we consider the variable set to \emph{true}. If the right button is not used, we consider the variable set to \emph{false}. A button not used in a variable is an \emph{available output}, and can then serve as an \emph{available input} to be used in other gadgets.

Every \emph{clause gadget} has its own column, with exactly four buttons. The topmost button (\emph{clause button}) is alone in its row; the others are inputs. If at least one of these is an available input, then we can match the clause button with all available inputs. We construct one clause gadget per formula clause, connecting its inputs to the appropriate variable outputs. Then, we can clear all the clauses just when we have made variable selections that satisfy the formula. 

The variables are connected to the clauses via a multi-purpose \emph{split gadget} (\cref{fig:2c_splitter}). Unlike the variable and the clause, this gadget uses buttons of two colors. The bottom button is an input; the top two are outputs. If the input button is available, we can match the middle row of the gadget as shown in \cref{fig:2c_splitter-1}, leaving the output buttons available. But if the input is not available, then the only way the middle row can be cleared is to first clear the red buttons in vertical pairs, as shown in \cref{fig:2c_splitter-2}; then the output buttons are not available.

\ifabstract
We provide a further description of the split gadget and complete the proof sketch in Appendix B.
\later{
\subsection*{Conclusion of Proof Sketch for Theorem \ref{thm:2colorsSAT}}
The split gadget has two additional effects beyond splitting. First, it changes the ``signal color'' from one of the two colors to the other. Second, it rotates the ``signal direction'' from vertical to horizontal: if the input is available to be used vertically by the split, then the outputs are available to be used horizontally as inputs to other gadgets. Therefore, we can restore the original signal color and direction by attaching further splits to the original split outputs. 

We split each variable output as often as needed to reach the clause inputs, each of which will also be the output of a split rotated 90\degree from \cref{fig:2c_splitter}. If we obtain more split outputs than clause inputs for a variable, then we add a red button in the same column as a red pair on an unused split branch. Then all buttons in this column can be cleared regardless of how the split is cleared. 
}
\fi

\iffull
The split gadget has two additional effects beyond splitting. First, it changes the ``signal color'' from one of the two colors to the other. Second, it rotates the ``signal direction'' from vertical to horizontal: if the input is available to be used vertically by the split, then the outputs are available to be used horizontally as inputs to other gadgets. Therefore, we can restore the original signal color and direction by attaching further splits to the original split outputs. 

We split each variable output as often as needed to reach the clause inputs, each of which will also be the output of a split rotated 90\degree from \cref{fig:2c_splitter}. If we obtain more split outputs than clause inputs for a variable, then we add a red button in the same column as a red pair on an unused split branch. Then all buttons in this column can be cleared regardless of how the split is cleared. 
\fi

\end{sketch}

\drieplaatjes [scale = .205] {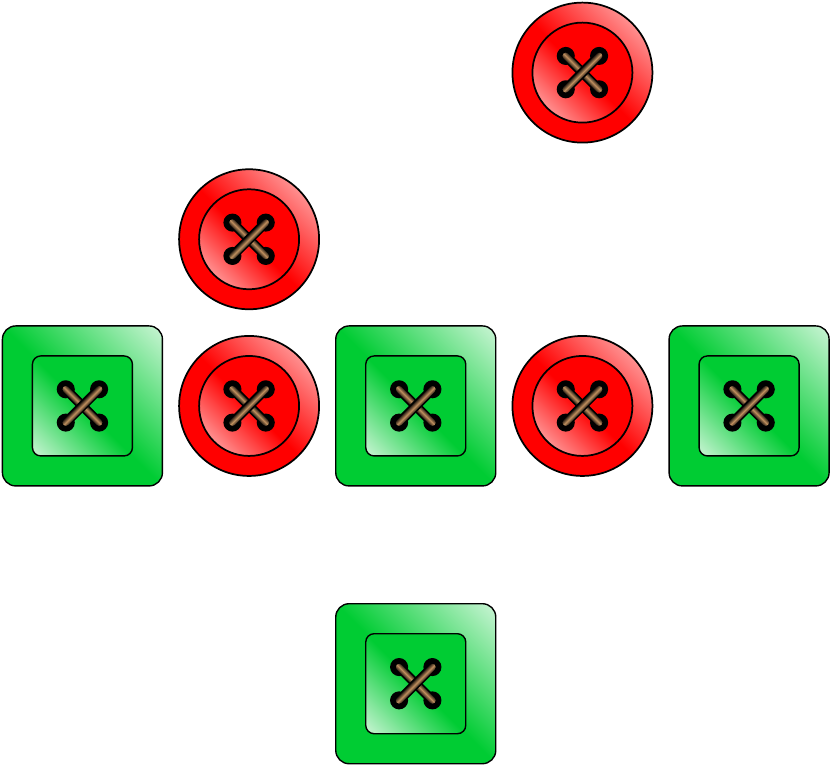} {2c_splitter-1} {2c_splitter-2} {Split gadget (a) and the two possible ways to clear it (b) and (c).}

\subsubsection {Polynomial-Time Algorithm for $1$-color and any cut directions.}
Given an instance $B$ with $C=1$ color and cut directions $d \in \{\dirsFour,\dirsThree,\dirsTwoA,\dirsOne\}$, we construct a hypergraph $G$ that has one node per button in $B$. 
A set of nodes is connected with a hyperedge if the corresponding buttons lie on the same horizontal, vertical, or diagonal line whose direction is in $d$, i.e., they 
can potentially be removed by the same cut. 
By Remark~\ref{rem:2or3cuts} it is sufficient to consider a hypergraph with only 2- and 3-edges. A solution to $B$ corresponds to a perfect matching in $G$. 
For clarity, we shall call a 3-edge in $G$ a \emph{triangle}, and a 2-edge simply an \emph{edge}.

Cornu\'{e}jols et al.~\cite{packing-subgraphs} showed how to compute a perfect $K_2$ and $K_3$ matching in a graph in polynomial time. However, their result is not directly applicable to our graph $G$ yet, as we need to find a matching that consists only of edges and proper triangles, and avoids $K_3$'s formed by cycles of three edges.

\begin{figure}[t]
\centering
\raisebox{-1\height}{\includegraphics[scale=0.64]{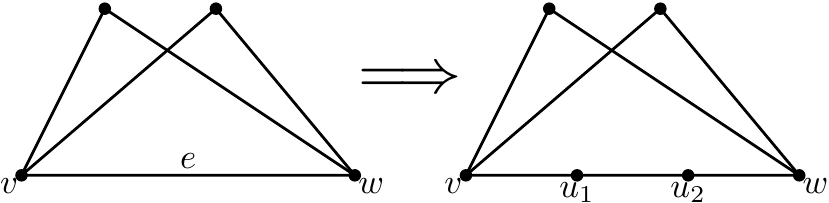}} 
\hfill
\raisebox{-1\height}{\includegraphics[scale=0.64]{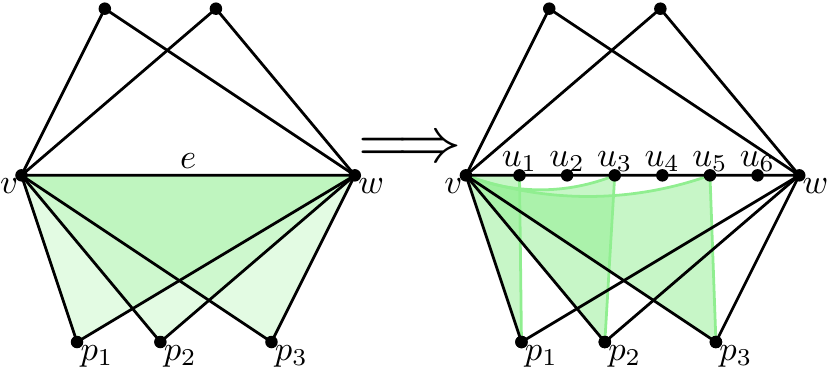}} 
\smallskip \\
\hrulefill \\
\smallskip
\comicII{0.20\textwidth}{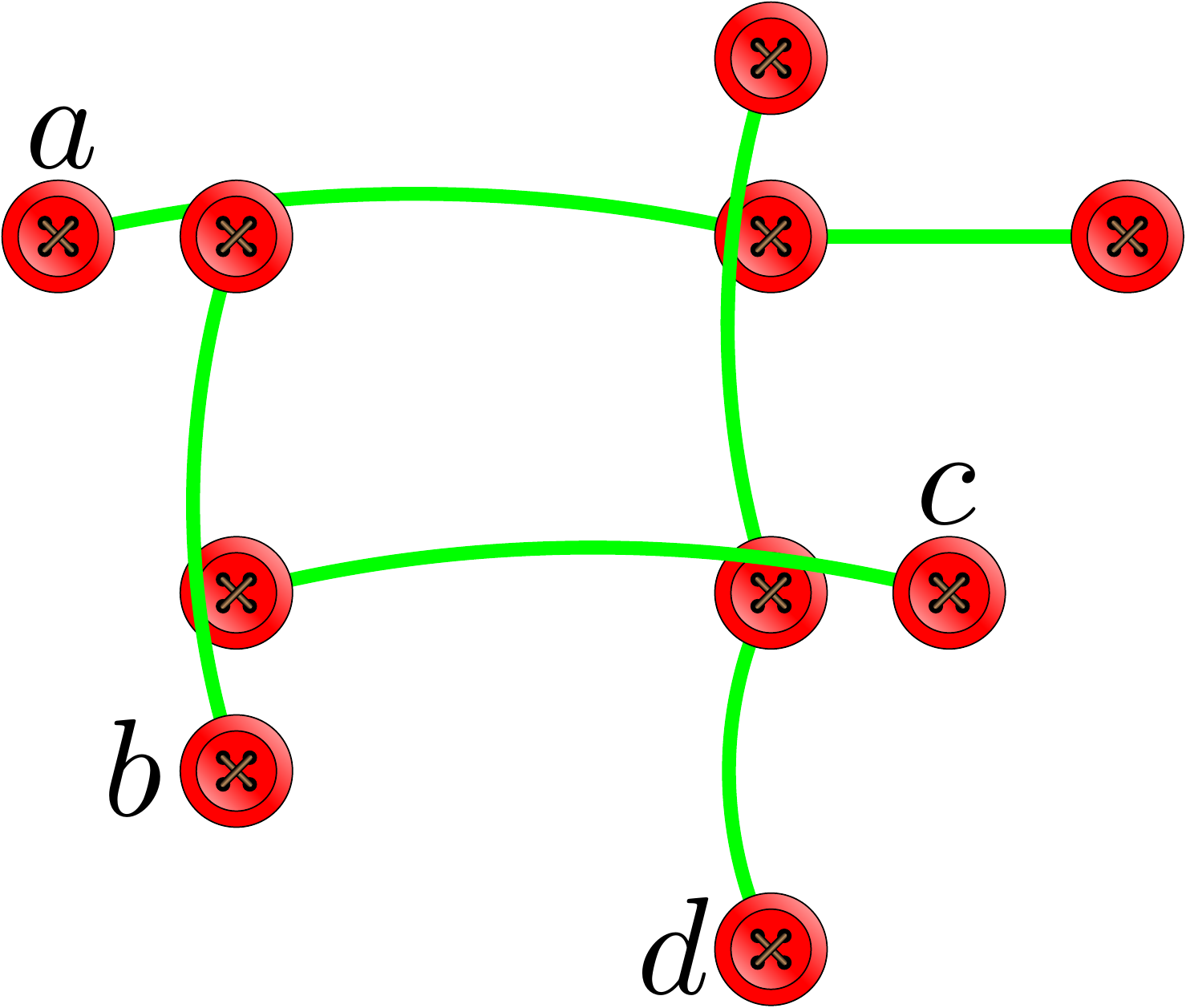} 
$\implies$
\comicII{0.1\textwidth}{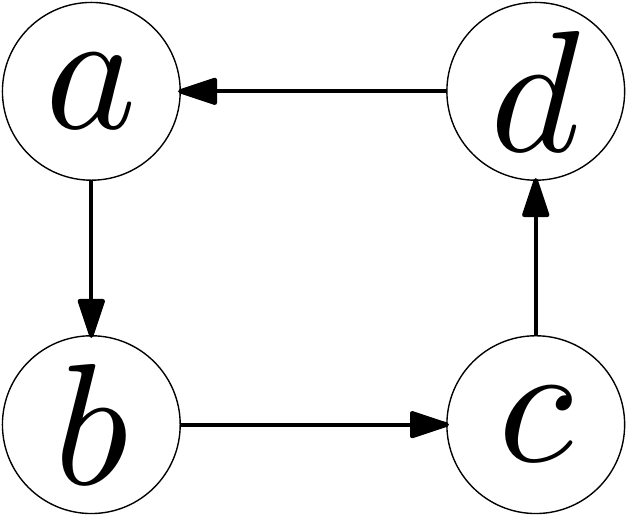} 
\hfill
\comicII{0.20\textwidth}{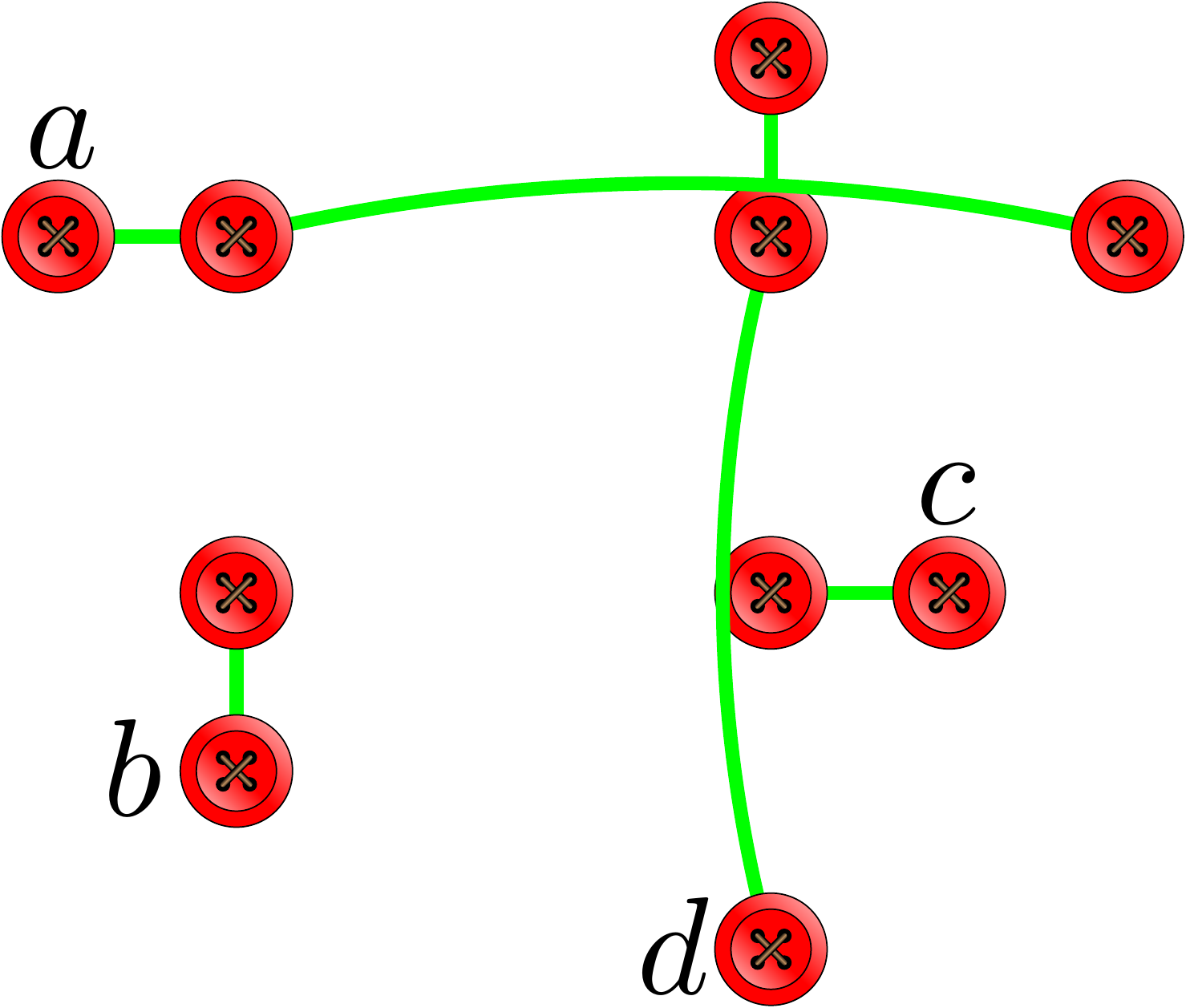} 
$\implies$
\comicII{0.1\textwidth}{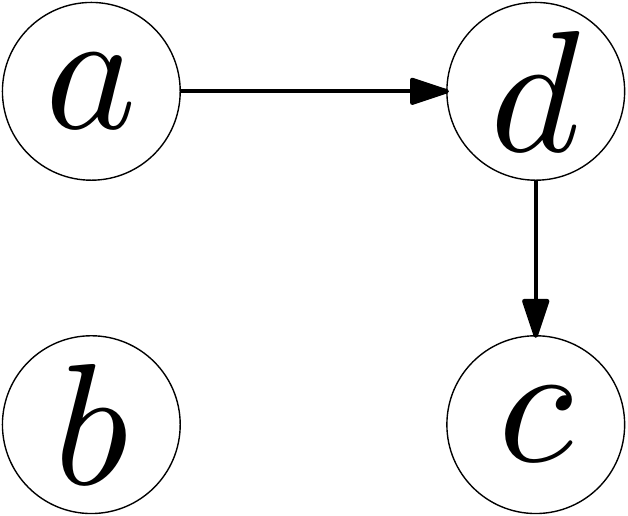} 
\caption{Top-left: splitting 3-cycles when there are no adjacent triangles to edge $e$; top-right: splitting 3-cycles when $e$ has adjacent triangles (shaded).
Bottom-left: constructing $G_c$ from four cuts blocking each other in a cycle; bottom-right: constructing $G_c$ from the same cuts after reassigning the blocking buttons
}
\label{fig:remove-graph-cycles}
\label{fig:remove-cycle}
\end{figure}


To apply \cite{packing-subgraphs} we construct graph $G'$ by adding vertices to eliminate all cycles of three edges as follows (see top of \cref{fig:remove-graph-cycles}). Start with $G'=G$. Consider an $e=(v,w)\in G'$ in a 3-cycle (a cycle of three edges). There are two cases: $e$ is not adjacent to any triangle in $G'$, or $e$ is adjacent to some triangles in $G'$. In the first case we add vertices $u_1$ and $u_2$ that split $e$ into three edges $(v,u_1)$, $(u_1,u_2)$, and $(u_2,w)$. In the second case, when $e$ is adjacent to $k$ triangles, we add $2k$ vertices ${u_1,u_2\dots,u_{2k}}$ along $e$, and replace every $\triangle p_{i}vw$ with $\triangle p_{i}vu_{2i-1}$.

\begin{lemma}\label{lem:remove-cycles}
There exists a perfect edge- and triangle-matching in $G'$ iff there exists perfect edge- and triangle-matching in $G$.
\end{lemma}
\begin{proof}
Given a perfect matching $M$ in $G$, we construct a perfect matching $M'$ in $G'$. 
Consider $e=(v,w)$ in $G$. If $e$ is not adjacent to any triangles in $G$, then
\begin{itemize}
	\item if $e\in M$ then add edges $(v,u_1)$ and $(u_2,w)$ of $G'$ to $M'$ (both $v$ and $w$ are covered by $e$, and all $v$, $w$, $u_1$, and $u_2$ are covered by $M'$);
	\item if $e\not\in M$ then add edge $(u_1,u_2)$ of $G'$ to $M'$ ($v$ and $w$ are not covered by $e$, and $u_1$ and $u_2$ are covered by $M'$).
\end{itemize}
In both cases above the extra nodes in $G'$ are covered by edges in $M'$, and if $v$ and $w$ in $G$ are covered by $e$ in $M$ then $v$ and $w$ are covered by $(v,u_1)$ and $(u_2,w)$ in $G'$. If $e$ is adjacent to some triangles in $G$,
\begin{itemize}
	\item if $e\in M$ then in $G'$ add edges $(v,u_1)$, $(u_{2k},w)$, and $(u_{2j},u_{2j+1})$ to $M'$, for $1\leq j<k$;
	\item if $\triangle p_{i}vw\in M$ for some $i$ then add $\triangle p_{i}vu_{2i-1}$, edges $(u_{2j-1},u_{2j})$ for $1\leq j<i$, $(u_{2j},u_{2j+1})$ for $i\leq j < k$, and $(u_{2k},w)$ of $G'$ to $M'$;  
	\item if neither $e$ nor any triangle adjacent to $e$ is in $M$ then add edges $(u_{2j-1},u_{2j})$ of $G'$ to $M'$, for $1\leq j \leq k$.
\end{itemize}
In all the above cases the extra nodes in $G'$ are covered by edges in $M'$, and if $v$ and $w$ in $G$ are covered by $e$ or a triangle in $M$ then $v$ and $w$ are also covered by $(v,u_1)$ and $(u_2,w)$ or by a corresponding triangle in $G'$.

\ifabstract
Refer to Appendix B for the details on how to create a perfect matching in $G$ from one in $G'$.
\fi
\iffull
Now, given a perfect matching $M'$ in $G'$ we show how to construct a perfect matching $M$ in $G$. Again, consider an edge $e$ in $G$ that is replaced by several edges in $G'$. If $e$ is not adjacent to any triangles in $G$, then
\begin{itemize}
	\item if $u_1$ is covered by edge $(v,u_1)$ in $M'$ then $u_2$ has to be covered by edge $(u_2,w)$, therefore we can add edge $e$ to $M$ (all $v$, $w$, $u_1$, and $u_2$ are covered by $M'$, and both $v$ and $w$ are covered by $e$ in $M$);
	\item if $u_1$ is covered by edge $(u_1,u_2)$ in $M'$ then $v$ and $w$ have to be covered by other edges or triangles in $M'$, and $v$ and $w$ will be covered by the corresponding edges or triangles in $M$.
\end{itemize}
In the second case, when $e$ is adjacent to some triangles in $G$,
\begin{itemize}
  \item if $\triangle p_{i}vu_{2i-1}\in M'$ for some $i$ then add $\triangle p_{i}vw$ to $M$. Edges $(u_{2j-1},u_{2j})$ for $1\leq j<i$, $(u_{2j},u_{2j+1})$ for $i\leq j< k$, and $(u_{2k},w)$ are forced in $M'$, and thus both nodes $v$ and $w$ are covered. Both nodes $v$ and $w$ are covered by the corresponding triangle $\triangle p_{i}vw$ in $M$;

	\item if none of the triangles $\triangle p_{i}vu_{2i-1}$ is in $M'$ then consider node $u_1$ in $G'$:
	\begin{itemize}
		\item if $u_1$ is covered by edge $(v,u_1)$ in $M'$ then edges $(u_{2j},u_{2j+1})$ for $1\leq j<k$ and $(u_{2k},w)$ are forced in $M'$. Therefore we can add edge $e$ to $M$ (all $v$, $w$, $u_j$ are covered by $M'$, and both $v$ and $w$ are covered by $e$ in $M$);
		\item if $u_1$ is covered by edge $(u_1,u_2)$ in $M'$ then edges $(u_{2j-1},u_{2j})$ for $2\leq j\leq k$ are forced in $M'$, therefore $v$ and $w$ have to be covered by other edges or triangles in $M'$, and $v$ and $w$ will be covered by the corresponding edges or triangles in $M$.
	\end{itemize}
\end{itemize}
\else
\later{
\subsection*{Part 2 of Proof of Lemma~\ref{lem:remove-cycles}}
Now, given a perfect matching $M'$ in $G'$ we show how to construct a perfect matching $M$ in $G$. Again, consider an edge $e$ in $G$ that is replaced by several edges in $G'$. If $e$ is not adjacent to any triangles in $G$, then
\begin{itemize}
	\item if $u_1$ is covered by edge $(v,u_1)$ in $M'$ then $u_2$ has to be covered by edge $(u_2,w)$, therefore we can add edge $e$ to $M$ (all $v$, $w$, $u_1$, and $u_2$ are covered by $M'$, and both $v$ and $w$ are covered by $e$ in $M$);
	\item if $u_1$ is covered by edge $(u_1,u_2)$ in $M'$ then $v$ and $w$ have to be covered by other edges or triangles in $M'$, and $v$ and $w$ will be covered by the corresponding edges or triangles in $M$.
\end{itemize}
In the second case, when $e$ is adjacent to some triangles in $G$,
\begin{itemize}
  \item if $\triangle p_{i}vu_{2i-1}\in M'$ for some $i$ then add $\triangle p_{i}vw$ to $M$. Edges $(u_{2j-1},u_{2j})$ for $1\leq j<i$, $(u_{2j},u_{2j+1})$ for $i\leq j< k$, and $(u_{2k},w)$ are forced in $M'$, and thus both nodes $v$ and $w$ are covered. Both nodes $v$ and $w$ are covered by the corresponding triangle $\triangle p_{i}vw$ in $M$;

	\item if none of the triangles $\triangle p_{i}vu_{2i-1}$ is in $M'$ then consider node $u_1$ in $G'$:
	\begin{itemize}
		\item if $u_1$ is covered by edge $(v,u_1)$ in $M'$ then edges $(u_{2j},u_{2j+1})$ for $1\leq j<k$ and $(u_{2k},w)$ are forced in $M'$. Therefore we can add edge $e$ to $M$ (all $v$, $w$, $u_j$ are covered by $M'$, and both $v$ and $w$ are covered by $e$ in $M$);
		\item if $u_1$ is covered by edge $(u_1,u_2)$ in $M'$ then edges $(u_{2j-1},u_{2j})$ for $2\leq j\leq k$ are forced in $M'$, therefore $v$ and $w$ have to be covered by other edges or triangles in $M'$, and $v$ and $w$ will be covered by the corresponding edges or triangles in $M$.
	\end{itemize}
\end{itemize}
\qed
}
\fi  \qed
\end{proof}

Thus, a perfect edge- and triangle-matching in $G$ that does not use a 3-cycle (if it exists) can be found by first converting $G$ to $G'$ and applying the result in \cite{packing-subgraphs} to $G'$. A solution of $B$ consisting of 2- and 3-cuts can be reduced to a perfect edge- and triangle-matching in $G$; however, the opposite is not a trivial task. A perfect matching in $G$ can correspond to a set of cuts $C_M$ in $B$ that are blocking each other (see bottom of \cref{fig:remove-cycle}). To extract a proper order of the cuts we build another graph $G_{c}$ that has a node per cut in $C_M$ and a directed edge between two nodes if the cut corresponding to the second node is blocking the cut corresponding to the first node. If $G_{c}$ does not have cycles, then there is a partial order on the cuts. The cuts that correspond to the nodes with no outgoing edges can be applied first, and the corresponding nodes can be removed from $G_{c}$. However, if $G_{c}$ contains cycles, there is no order in which the cuts can be applied to clear up board $B$. In this case we will need to modify some of the cuts in order to remove cycles from $G_c$. \ifabstract We provide the details in Appendix B.
\later{
\subsection*{Removing Cycles from $G_c$}
To do so we perform the following two steps:
\begin{itemize}
	\item Repeat: choose a cycle in $G_c$ and for every edge $(c_1,c_2)$ in it reassign the button of $c_2$ that is blocking cut $c_1$ to $c_1$. After one step
\begin{itemize}
	\item the total number of buttons in all the cuts stays the same,
	\item if cut $c_2$ consisted of two buttons, or if cut $c_2$ consisted of three buttons and the button blocking cut $c_1$ was not in the middle of $c_2$, then after reassigning the buttons the length (i.e., the Euclidean distance between the buttons) of $c_2$ decreases,
	\item if cut $c_2$ consisted of three buttons and the button blocking cut $c_1$ was in the middle of $c_2$, then after reassigning the buttons the direction of the edge $(c_1,c_2)$ changes to $(c_2,c_1)$.
\end{itemize}
	\item After the previous step can no longer be applied to $G_c$ such that the total length of the cuts decreases, there can only be cycles left in $G_c$ that consist of 3-cuts with the blocking buttons being the middle ones. Then for any edge $(c_1,c_2)$, if we reassign the middle button of $c_2$ to $c_1$, $c_2$ will become a 2-cut, and $c_1$ will have four buttons, and therefore can be split into two 2-cuts. The direction of the corresponding edge will also change its direction. In this way the rest of the cycles can be removed from $G_c$.
\end{itemize}
}
\fi

\iffull
To do so we perform the following two steps:
\begin{itemize}
	\item Repeat: choose a cycle in $G_c$ and for every edge $(c_1,c_2)$ in it reassign the button of $c_2$ that is blocking cut $c_1$ to $c_1$. After one step
\begin{itemize}
	\item the total number of buttons in all the cuts stays the same,
	\item if cut $c_2$ consisted of two buttons, or if cut $c_2$ consisted of three buttons and the button blocking cut $c_1$ was not in the middle of $c_2$, then after reassigning the buttons the length of $c_2$ decreases,
	\item if cut $c_2$ consisted of three buttons and the button blocking cut $c_1$ was in the middle of $c_2$, then after reassigning the buttons the direction of the edge $(c_1,c_2)$ changes to $(c_2,c_1)$.
\end{itemize}
	\item After the previous step can no longer be applied to $G_c$ such that the total length of the cuts decreases, there can only be cycles left in $G_c$ that consist of 3-cuts with the blocking buttons being the middle ones. Then for any edge $(c_1,c_2)$, if we reassign the middle button of $c_2$ to $c_1$, $c_2$ will become a 2-cut, and $c_1$ will have four buttons, and therefore can be split into two 2-cuts. The direction of the corresponding edge will also change its direction. In this way the rest of the cycles can be removed from $G_c$.
\end{itemize}
\fi
\iffull

To summarize, a solution to 1-color Buttons $\&$ Scissors level $B$ can be found, if it exists, with following algorithm:
\begin{enumerate}
	\item Convert $B$ to hypergraph $G$ that encodes all possible cuts of length two and three buttons using the allowed cut directions.
	\item Convert $G$ to $G'$ that contains no 3-cycles that are not triangles, and find a matching in $G'$.
	\item Construct the directed graph $G_c$ that encodes which cuts from the matching are blocking each other and remove all the cycles from $G_c$ by reassigning some buttons to other cuts.
	\item Extract a partial order from $G_c$ that will give a proper order in which the cuts can be applied to solve $B$.
\end{enumerate}
\fi
By Lemma~\ref{lem:remove-cycles} and by the construction above we obtain the following theorem.

\ifabstract
\begin{theorem}
$\Board{n \times n}{1}{\infty}{d}{\{2,3\}}{B}$ is polytime solvable for all $d \in \{\dirsFour,\dirsThree,\dirsTwoA,\dirsOne\}$.
\end{theorem}
\fi
\iffull
\begin{theorem}
Buttons $\&$ Scissors for 1-color, i.e., $\Board{n \times n}{1}{\infty}{\dirsFour}{\{2,3\}}{B}$, is polytime solvable.
\end{theorem}
\fi


\subsection{Frequency of Colors} \label{sec:1p_frequency} 


\ifabstract
\begin{theorem} \label{thm:freq3}
$\Board{n \times n}{\infty}{3}{\dirsFour}{\{2,3\}}{B}$ is polytime solvable.
\end{theorem}
\fi
\iffull
\begin{theorem} \label{thm:freq3}
Buttons $\&$ Scissors with maximum color frequency $F=3$, i.e. $\Board{n \times n}{\infty}{3}{\dirsFour}{\{2,3\}}{B}$, is polytime solvable.
\end{theorem}
\fi

\begin{proof}
A single cut in any solution removes a color.
By \cref{rem:removeColor}, these cuts do not make a solvable board unsolvable.
Thus, a greedy algorithm suffices.  \qed
\end{proof}

Hardness was established for maximum frequency $F=7$ in \cite{glsw-bsnpc-15}.
We strengthen this to $F=4$ via the modified clause gadget in \cref{fig:LevelSeven} (c).
In this gadget the leftmost circular button can be removed if and only if at least one of the three non-circular buttons is removed by a vertical cut.
Thus, it can replace the clause gadget in Section 4.1 of \cite{glsw-bsnpc-15}.
\ifabstract
\cref{thm:freq4} is proven in the appendix.
\fi

\ifabstract
\begin{theorem} \label{thm:freq4}
$\Board{n \times n}{\infty}{4}{\dirsFour}{\{2,3\}}{B}$ is NP-complete.
\end{theorem}
\fi
\iffull
\begin{theorem} \label{thm:freq4}
Buttons $\&$ Scissors with maximum color frequency $F=4$, i.e. $\Board{n \times n}{\infty}{4}{\dirsFour}{\{2,3\}}{B}$, is NP-complete.
\end{theorem}
\fi

\ifabstract
\later{
\subsection*{Proof of \cref{thm:freq4}}
\fi

\begin{proof}
The puzzle with $F=7$, i.e. $\Board{n \times n}{\infty}{7}{\dirsFour}{2}{B}$, was proven NP-complete in \cite{glsw-bsnpc-15} via 3-SAT whose clauses have literals of distinct variables.
(The construction created boards that do not have buttons of the same color on any diagonal, so the NP-completeness was also established for these parameters and both $d = \dirsThree$ and $d = \dirsTwoA$ in \cite{glsw-bsnpc-15} .)

We obtain hardness for $F=4$ by modifying the original proof's clause gadget.
The top of \cref{fig:OR_JCDCGG} illustrates the original gadget for clause $C_x = L_i \lor L_j \lor L_k$, where $L_i$ is either the positive literal for variable $V_i$ or the negative literal $\lnot V_i$, and similarly $L_j \in \{V_j, \lnot V_j\}$ and $L_k \in \{V_k, \lnot V_k\}$.
Included in this gadget are the following buttons:
\begin{itemize}
  \item Two clause buttons $\CButt[L]{x}$ and $\CButt[R]{x}$.
  The $L$ and $R$ subscripts denote \underline{L}eft and \underline{R}ight, respectively.
  \item Literal instance buttons $\VCButt[M]{i,x}$, $\VCButt[M]{j,x}$, and $\VCButt[M]{k,x}$.
  The parameter $\VCButt{,x}$ refers to the clause $C_x$ under consideration, and we omit it from \cref{fig:OR_JCDCGG} and the discussion below.  
  The $M$ subscripts denote \underline{M}iddle.
\end{itemize}
By convention different shapes and interior labels different colored buttons, whereas subscripts do not alter the color.
Thus, $\CButt[L]{x}$ and $\CButt[R]{x}$ have the same color, whereas $\VCButt[M]{i}$, $\VCButt[M]{j}$, and $\VCButt[M]{k}$ are all distinct.
Each of the middle buttons can be removed by a vertical cut as denoted by the downward arrows.
This original clause gadget has the following property: 
$\CButt[L]{x}$ can be removed if and only if at least one of $\VCButt[M]{i}$, $\VCButt[M]{j}$, or $\VCButt[M]{k}$ is removed by vertical cut.

The bottom of \cref{fig:OR_JCDCGG} illustrates the new gadget for the same clause $C_x = L_i \lor L_j \lor L_k$.
If $n$ is the number of variables in the 3-SAT instance, then original gadget was contained in one row and $4n+12$ columns, including $n+1$ blank columns to clarify the presentation.
The new gadget uses $8n+5$ rows and $8n+5$ columns (see the bottom of Figure 2 in \cite{glsw-bsnpc-15}).
In this gadget we add two buttons to those discussed above:
\begin{itemize}
  \item Two additional clause buttons $\CButt[T]{x}$ and $\CButt[B]{x}$.
  The $T$ and $B$ subscripts denote \underline{T}op and \underline{B}ottom, respectively.
\end{itemize}

We now prove that the new gadget has the same property as the old gadget.
We consider four cases that cover all possibilities:
\begin{itemize}
  \item If $\VCButt[M]{i}$ is cut vertically, then cut $\CButt[L]{x}$ and $\CButt[T]{x}$ (and then $\CButt[B]{x}$ and $\CButt[R]{x}$);
  \item If $\VCButt[M]{j}$ is cut vertically, then cut $\CButt[L]{x}$ and $\CButt[B]{x}$ (and then $\CButt[T]{x}$ and $\CButt[R]{x}$);
  \item If $\VCButt[M]{k}$ is cut vertically, then cut $\CButt[L]{x}$ and $\CButt[R]{x}$ (and then $\CButt[T]{x}$ and $\CButt[B]{x}$);
  \item If none of these buttons is removed by a vertical cut, then $\CButt[L]{x}$ cannot be removed since the only other buttons of the same color are blocked.
\end{itemize}
Therefore, the new gadget is a perfect replacement for the previous clause gadget.

This substitution increases the frequency of each clause button from $2$ to $4$, and decreases the frequency of each instance literal button color to $3$.
Thus, the maximum frequency is reduced from $F=7$ to $F=4$ by Remark 6 in \cite{glsw-bsnpc-15}.  \qed
\end{proof}

\begin{figure}
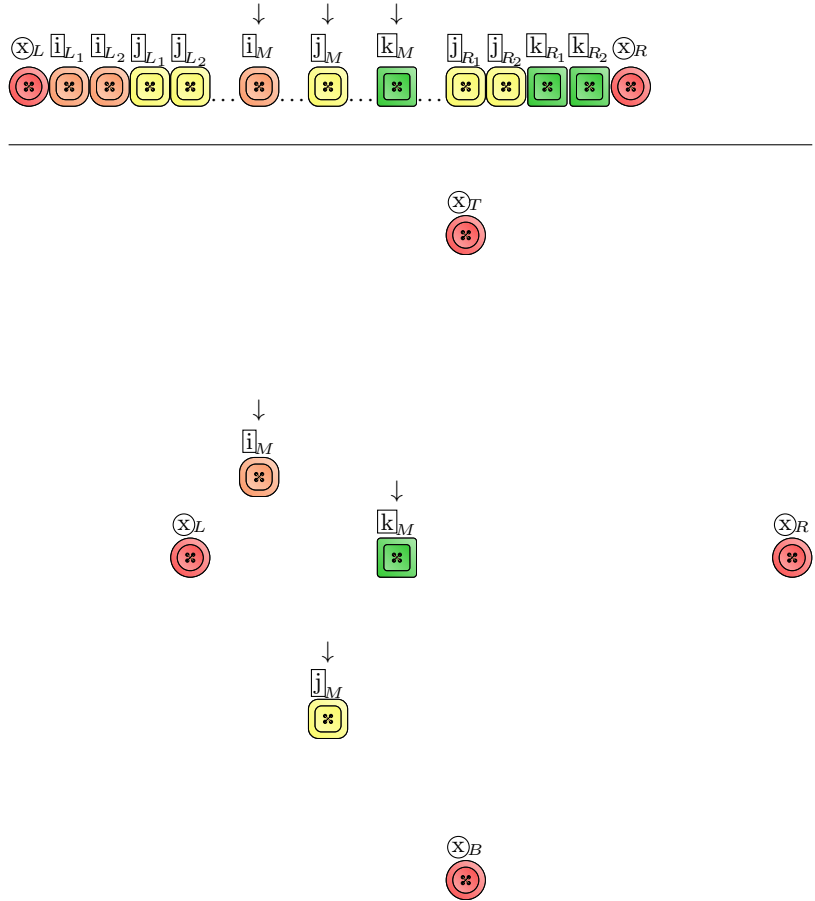

\begin{tabular}{*{21}{@{}c@{}}}
& & & & & & $\downarrow$ & & $\downarrow$ & & $\downarrow$\ & & & & & & & & & & \\[0.2em]
$\CButt[L]{x}$ & $\VCButt[L_1]{i}$ & $\VCButt[L_2]{i}$ & $\VCButt[L_1]{j}$ & $\VCButt[L_2]{j}$ && $\VCButt[M]{i}$ && $\VCButt[M]{j}$ && $\VCButt[M]{k}$ && $\VCButt[R_{\!1}]{j}$ & $\VCButt[R_{\!2}]{j}$ & $\VCButt[R_{\!1}]{k}$ & $\VCButt[R_{\!2}]{k}$ & $\CButt[R]{x}$ \\
\BT1 & \BT2 & \BT2 & \BT3 & \BT3 & $\cdots$ & \BT2 & $\cdots$ & \BT3 & $\cdots$ & \BT4 & $\cdots$ & \BT3 & \BT3 & \BT4 & \BT4 & \BT1 & \BT0 & \BT0 & \BT0 & \BT0 \\ \\
\hline \\ 
& & & & & & & & & & & & \raisebox{0.3em}{$\CButt[T]{x}$} & & & & & & & & \BT0 \\[-\dp\strutbox]
& & & & & & & & & & & & \BT1 & & & & & & & & \BT0 \\[-\dp\strutbox]
& & & & & & & & & & & & & & & & & & & & \BT0 \\[-\dp\strutbox]
& & & & & & & & & & & & & & & & & & & & \BT0 \\[-\dp\strutbox]
& & & & & & & & & & & & & & & & & & & & \BT0 \\[-\dp\strutbox]
& & & & & & $\downarrow$ & & & & & & & & & & & & & & \BT0 \\[-\dp\strutbox]
& & & & & & \raisebox{0.3em}{$\VCButt[M]{i}$} & & & & & & & & & & & & & & \BT0 \\[-\dp\strutbox]
& & & & & & \BT2 & & & & $\downarrow$ & & & & & & & & & & \BT0 \\[-\dp\strutbox]
& & & & \raisebox{0.3em}{$\CButt[L]{x}$} & & & & & & \raisebox{0.3em}{$\VCButt[M]{k}$} & & & & & & & & & \BT0 & \raisebox{0.3em}{$\CButt[R]{x}$} \\[-\dp\strutbox]
& & & & \BT1 & & & & & & \BT4 & & & & & & & & & & \BT1 \\[-\dp\strutbox]
& & & & & & & & & & & & & & & & & & & & \BT0 \\[-\dp\strutbox]
& & & & & & & & $\downarrow$ & & & & & & & & & & & & \BT0 \\[-\dp\strutbox]
& & & & & & & & \raisebox{0.4em}{$\VCButt[M]{j}$} & & & & & & & & & & & & \BT0 \\[-\dp\strutbox]
& & & & & & & & \BT3 & & & & & & & & & & & & \BT0 \\[-\dp\strutbox]
& & & & & & & & & & & & & & & & & & & & \BT0 \\[-\dp\strutbox]
& & & & & & & & & & & & & & & & & & & & \BT0 \\[-\dp\strutbox]
& & & & & & & & & & & & \raisebox{0.3em}{$\CButt[B]{x}$} & & & & & & & & \BT0 \\[-\dp\strutbox]
& & & & & & & & & & & & \BT1 & & & & & & & & \BT0 \\[-\dp\strutbox]
\end{tabular}
\caption{The original (top) and new (bottom) gadget for clause $C_x = L_1 \lor L_2 \lor L_3$.
In both cases the leftmost circular button $\CButt[L]{x}$ can be removed if and only if at least one of the three middle buttons $\VCButt[M]{i}$, $\VCButt[M]{j}$, $\VCButt[M]{k}$ is removed by a vertical cut, as denoted by downward arrows.
Labels $\VCButt{,x}$ are omitted from all of the square buttons to save space, and the $\cdots$ denote empty squares between the three middle buttons.
}
\label{fig:OR_JCDCGG}
\end{figure}

\ifabstract
} 
\fi


\subsection{Cut Sizes} \label{sec:1p_lengths}

\cref{sec:1p_colors} provided a polytime algorithm for 1-color.
However, if we reduce the cut size set from $\{2,3,4\}$ to $\{3,4\}$ then it is NP-complete.
We also strengthen \cref{thm:2colorsSAT} by showing that 2-color puzzles are hard with cut size set $\{2\}$.

\subsubsection{Hardness for Cut Sizes $\{3,4\}$ and 1-Color}

\begin{theorem}\label{thm:1colorPlanarSAT}
\ifabstract
$\Board{n \times n}{1}{\infty}{\dirsFour}{\{3,4\}}{B}$ is NP-complete.
\fi
\iffull
Buttons $\&$ Scissors for 1-color with cut sizes $\{3,4\}$, i.e. $\Board{n \times n}{1}{\infty}{\dirsFour}{\{3,4\}}{B}$, is NP-complete.
\fi
\end{theorem}

\begin{proof}

\iffull
\begin{figure}[b]
       \comic{.3\textwidth}{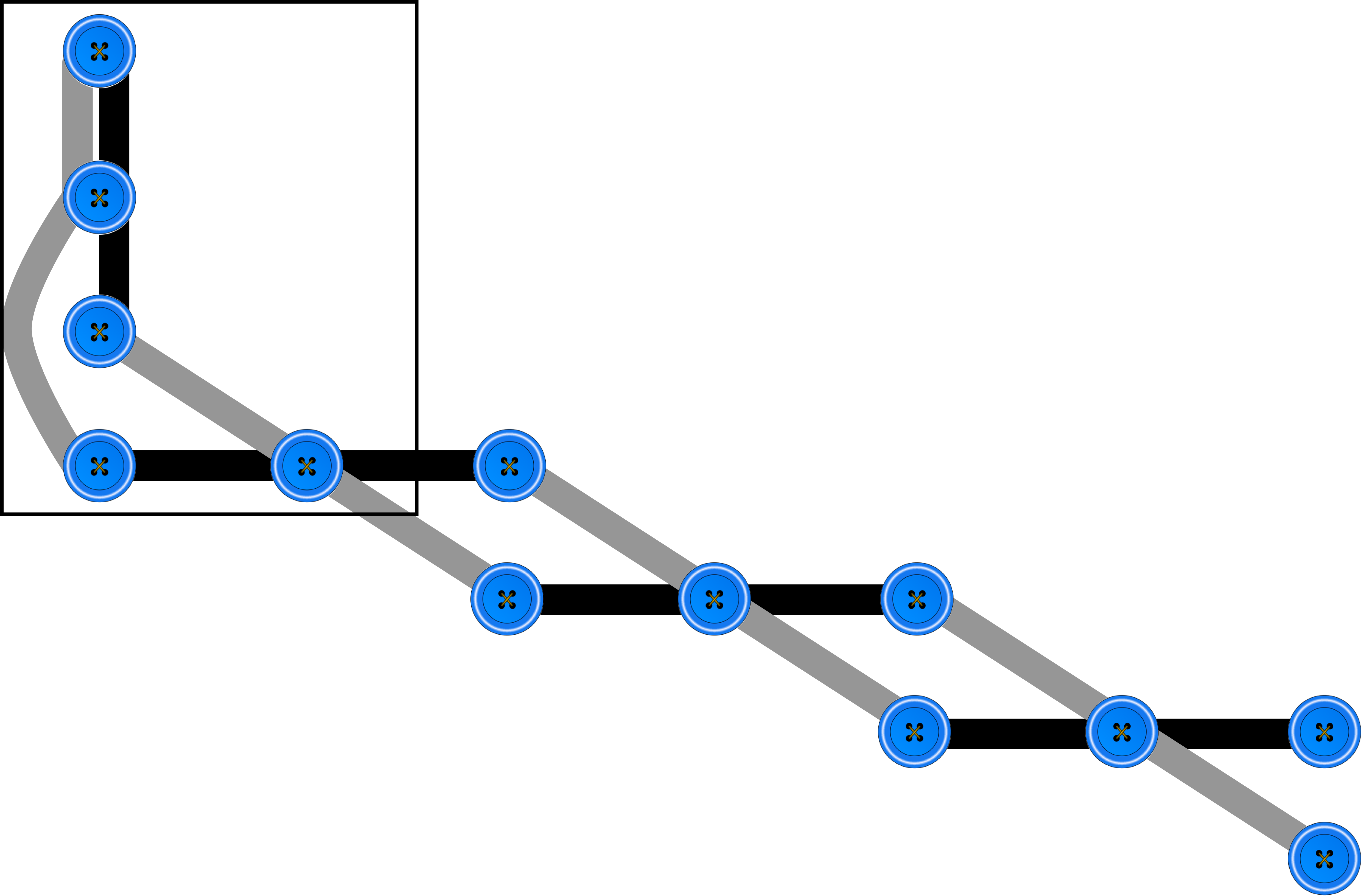}{(a)} 
       \hfill
       \comic{.6\textwidth}{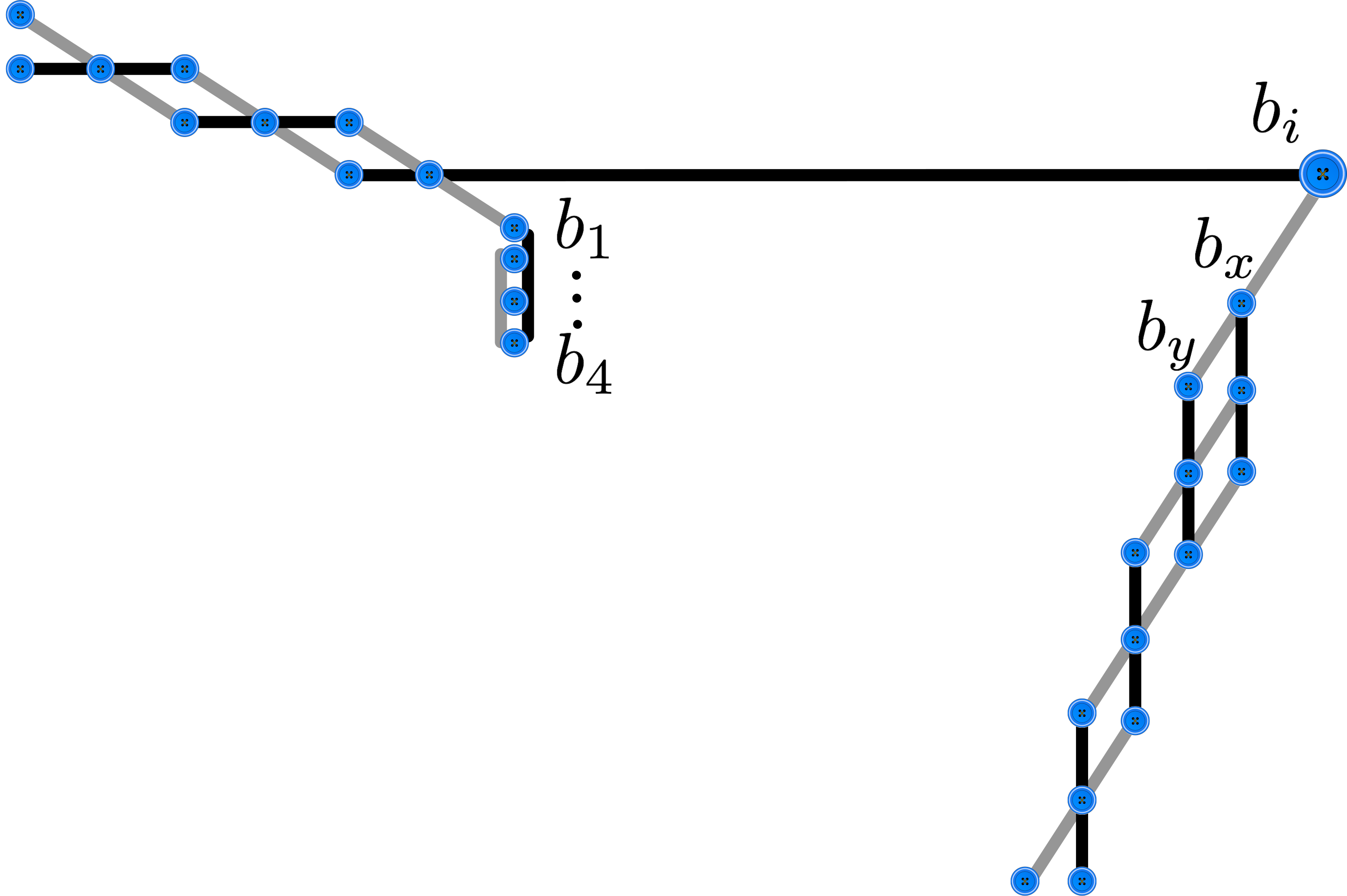}{(b)} 
         \caption{\label{fig:var-1p}\small (a) The only two cut possibilities in the variable gadget (shown in black and gray), corresponding to truth assignments of ``true'' and ``false'', respectively. (b) The bend gadget for the 1-color case.
         }
\end{figure}
\fi

\ifabstract
\begin{figure}[b]
       \comic{.25\textwidth}{cuts-geq-3-variable-2-bw}{(a)} 
       \hfill
       \comic{.53\textwidth}{cuts-geq-3-bend-bw} 
              {\vspace*{-2.5cm}\hspace*{1cm}(b)}\vspace*{-0.5cm}\\
       \centering\hspace*{-1cm}
       \comic{.58\textwidth}{cuts-geq-3-clause-h-s}{(c)} 
       \caption{\label{fig:var-1p}\small (a) The only two cut possibilities in the variable gadget (shown in black and gray), corresponding to truth assignments of ``true'' and ``false'', respectively. (b) The bend gadget for the 1-color case. (c) The clause gadget for the 1-color case.
       }
\end{figure}
\fi

We show $\Board{n \times n}{1}{\infty}{\dirsFour}{\{3,4\}}{B}$ to be NP-hard by a reduction from PLANAR 3-SAT, which was shown to be NP-complete by Lichtenstein~\cite{l-pftu-82}.

An instance $F$ of the PLANAR 3-SAT problem is a Boolean formula in 3-CNF consisting of a set \iffull$\mathcal{C} = \{C_1, C_2, \dots, C_m\}$\fi\ifabstract$\mathcal{C}$\fi\ of $m$ clauses over $n$ variables \iffull$\mathcal{V} = \{x_1, x_2, \dots, x_n\}$\fi\ifabstract$\mathcal{V}$\fi. 
\iffull
Clauses in $F$ contain variables and negated variables, denoted as \emph{literals}.
\fi
The variable-clause incidence graph $G=(\mathcal{C}\cup\mathcal{V}, E)$ is planar, \iffull where 
$\{C_i, x_j\} \in E \Leftrightarrow x_j$ or $\neg x_j$ is in $C_i$, and with all variables connected in a cycle.\fi 
\ifabstract and all variables are connected in a cycle.\fi\ 
\iffull
It is sufficient to consider formulae where $G$ has a rectilinear embedding, see Knuth and Raghunathan~\cite{kr-pcr-92}.
\fi
The PLANAR 3-SAT problem is to decide whether there exists a truth assignment to the variables such that at least one literal per clause is true.

We turn the planar embedding of $G$ into a Buttons \& Scissors board, i.e., 
we present 
variables, clauses and edges by single-color buttons 
that need to be cut.

\ifabstract
  We provide more detailed descriptions of each gadget in Appendix B.
\later{
  \subsection*{Description of Gadgets for Proof of \cref{thm:1colorPlanarSAT}}
  \begin{itemize}
    \item \textbf{Variable gadget} (\cref{fig:var-1p}(a)):
    The buttons to the right are positioned such that they can only be cut with either horizontal cuts of size 3 (black) or diagonal cuts of size 3 (gray). 
  To obtain a feasible cutting pattern for the rest, this leaves only two cut possibilities for the four vertically aligned buttons of the variable gadget: cutting the topmost three buttons, in which case further cuts in the wire are enforced to be horizontal, or cutting the two topmost and the bottom button (which is possible by executing the first diagonal cut first), in which case further cuts in the wire are enforced to be diagonal. These exactly two feasible solutions of the variable gadget correspond to a truth setting of  ``true'' (black) and ``false'' (gray) of the variable.

    \item \textbf{Bend gadget} (\cref{fig:var-1p}(b)):
      \iffull Note that the three lower of the four vertically aligned buttons, , $b_1,\ldots,b_4$,, have to be removed with the same vertical cut, as none of them is aligned with any other button pair on the board.
      \else Note that the three lower of the four vertically aligned buttons, $b_1,\ldots,b_4$, are not aligned with any other button pair on $B$, and so they must belong to the same vertical cut. 
      \fi
      
      The isolated button to the right (shown enlarged), $b_i$, has to be cut with each feasible cut pattern. 
  Entering the bend gadget with black, horizontal cuts, the last horizontal row features only two buttons, and so $b_i$ needs to be added to this cut to make it feasible. 
  Buttons $b_1,\ldots, b_4$ are all cut together in this case. 
  In the lower wire, $b_x$ and $b_y$ cannot be cut together (by cut sizes), enforcing further cuts to be vertical (again black).

	  \iffull
	  Entering the bend gadget with the gray, diagonal pattern, the last diagonal cut needs to use one of the four vertically aligned buttons, as otherwise the cut size would be only two. The remaining three of these buttons are deleted with one cut. This still leaves the isolated button to be cut, and of the buttons in the lower wire only the two buttons on the upper diagonal are aligned with it, enforcing further cuts to be diagonal. 
      \else A similar argument is used for the gray cut pattern.
      \fi

    \item \textbf{Split gadget} (\cref{fig:split-1p}(b)):
  The last diagonal cut (gray) and the last horizontal cut (black) of the input wire from the left only feature two buttons in this wire, i.e., if we want to make a feasible cut, we need to pick up another button, a single button that starts the two output wires. If we enter with the black cut pattern, the upper output wire's single button is cut, enforcing black, horizontal cuts in the rest. In the lower output wire, the single button still needs to be cut. Except for the two buttons of the last diagonal of the input wire, it is only aligned with the other uppermost buttons of this output wire, thus, the horizontal cuts need to be made, enforcing black, horizontal cuts in the rest of the wire. An analogous argument is used for the gray cut pattern.

  \item \textbf{Not gadget} (\cref{fig:split-1p}(a)):
  The last diagonal cut (gray) and the last horizontal cut (black) of the input wire from the left only feature two buttons in this wire. 
  Moreover, there is again a single isolated button. If the gadget is entered with the black, vertical cut pattern, the isolated button cannot be cut with any cut from the input wire. The only remaining possibility is to cut the isolated button with the two leftmost diagonal buttons of the output wire, enforcing diagonal cuts in the rest of the wire. If the gadget is entered with the gray, diagonal cut pattern, the isolated button needs to be cut with the last diagonal wire cut (as otherwise its size of two would render it infeasible). In the output wire the first diagonal cut would be to short, enforcing horizontal cuts in the rest of the wire.

  \item \textbf{Clause gadget} (\cref{fig:var-1p}(c)):
  Wires from literals approach the gadget along three diagonals. The clause has one isolated button (shown enlarged), and two additional buttons for each wire that build a row of four aligned buttons with the last two buttons of the wire. When a gray cut direction is used, the two lower plus the topmost of this row of four combine for a feasible cut. If a black cut direction is used, the three lower buttons of this row allow a feasible cut. But this still leaves the isolated button to cut: Only if in at least one of the input wires the black cuts, corresponding to a truth setting of ``true'' for the literal, are used, this button can be added to the last black cut of the wire (extending its size from 3 to 4). If all literals are set to ``false'', i.e., along all three input wires the gray cut pattern is used, cutting the isolated button with two others along one of the horizontal/vertical axis, leaves at least one isolated button, rendering the complete board infeasible, as we are not left with an empty board. 
  \end{itemize}
} 
\else
  We continue by describing each of the gadgets necessary.
\fi

\ifabstract

The {\bf variable gadget}, shown in \cref{fig:var-1p}(a), enables us to associate horizontal and diagonal cut patterns with ``true'' and ``false'' values, respectively.



\else
The {\bf variable gadget} is shown in \cref{fig:var-1p}(a). The buttons to the right are positioned such that they can only be cut with either horizontal cuts of size 3 (black) or diagonal cuts of size 3 (gray). , see \cref{fig:var-1p}(b). 
 To obtain a feasible cutting pattern for the rest, this leaves only two cut possibilities for the four vertically aligned buttons of the variable gadget: cutting the topmost three buttons, in which case further cuts in the wire are enforced to be horizontal, or cutting the two topmost and the bottom button (which is possible by executing the first diagonal cut first), in which case further cuts in the wire are enforced to be diagonal. These exactly two feasible solutions of the variable gadget correspond to a truth setting of  ``true'' (black) and ``false'' (gray) of the variable.
\fi

\ifabstract
The {\bf bend gadget}, shown in \cref{fig:var-1p}(b), enables us to bend a wire to match the bends in $G$'s embedding while enforcing that the same values are propagated through the bent wire.


\else
The {\bf bend gadget}, shown in \cref{fig:var-1p}(b), enables us to bend a wire to match the bends in $G$'s embedding while enforcing that the same 
 values are propagated through the bent wire. 
 \iffull Note that the three lower of the four vertically aligned buttons, , $b_1,\ldots,b_4$,, have to be removed with the same vertical cut, as none of them is aligned with any other button pair on the board.\fi
\fi

\begin{figure}
\center
	\hfill
    \comic{.6\textwidth}{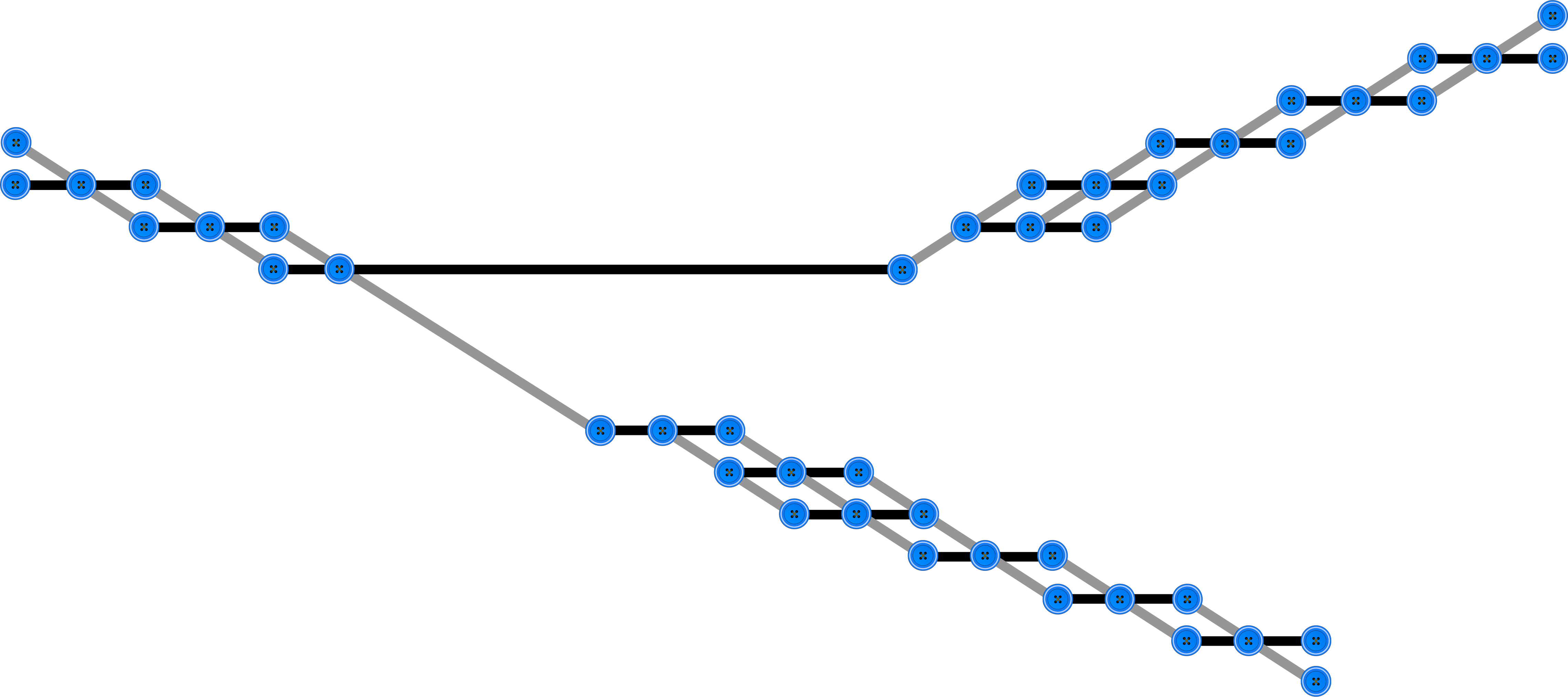}{\hspace*{.5cm}(b)}\\ 
    \vspace*{-3.2cm}
    \hspace*{-.4\textwidth}
    \comic{.5\textwidth}{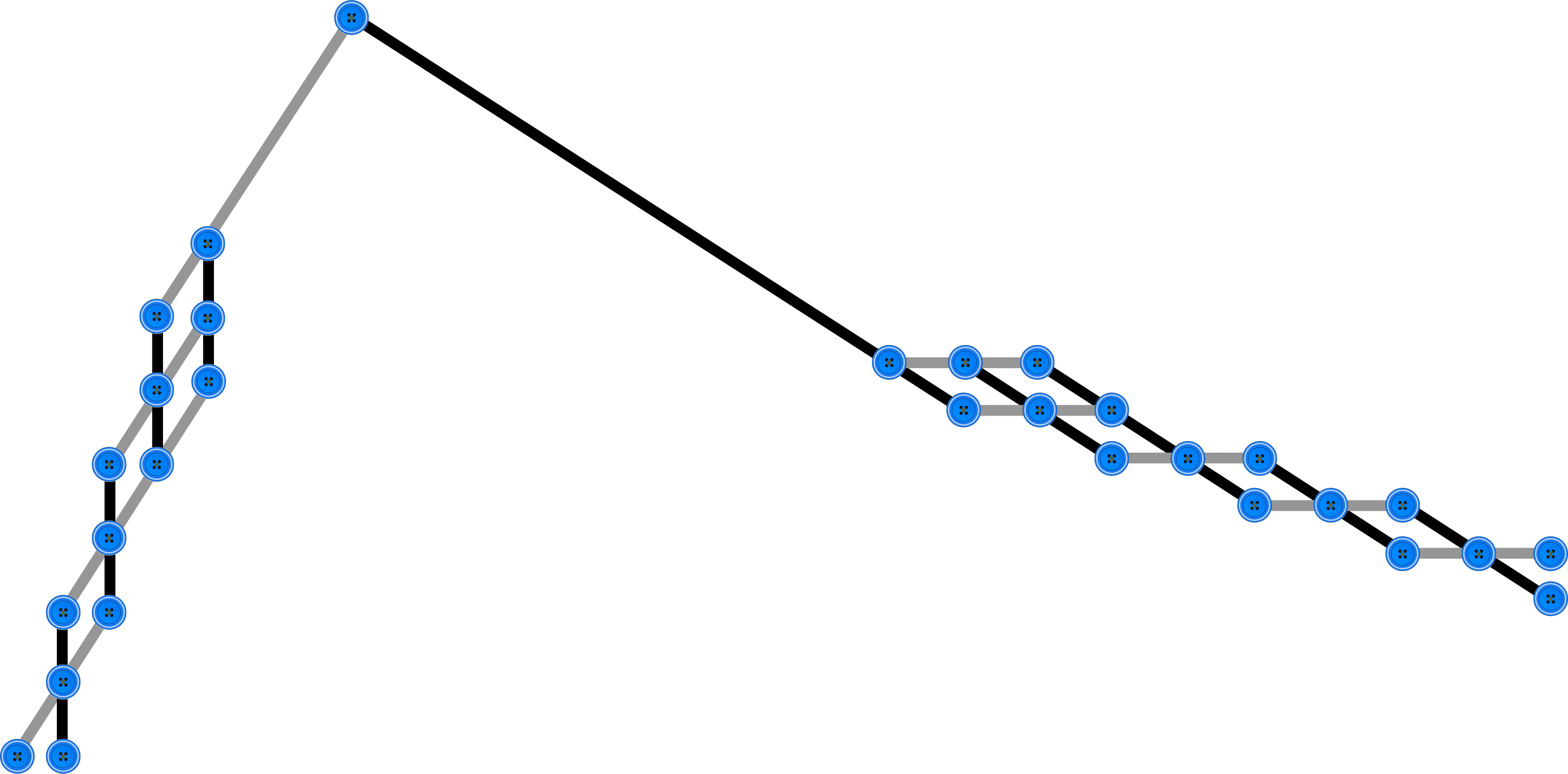}{\vspace*{-.3cm}(a)} 
         \caption{\label{fig:split-1p}\small (a) The not gadget, negating the input truth assignment, for the 1-color case. (b) The split gadget for the 1-color case.
         }
\end{figure}

\ifabstract
The {\bf split gadget}, shown in \cref{fig:split-1p}(b), enables us to increase the number of wires leaving a variable and propagating its truth assignment.
\fi


\iffull
The {\bf split gadget}, shown in \cref{fig:split-1p}(b), enables us to increase the number of wires leaving a variable and propagating its truth assignment.
The last diagonal cut (gray) and the last horizontal cut (black) of the input wire from the left only feature two buttons in this wire, i.e., if we want to make a feasible cut, we need to pick up another button, a single button that starts the two output wires. If we enter with the black cut pattern, the upper output wire's single button is cut, enforcing black, horizontal cuts in the rest. In the lower output wire, the single button still needs to be cut. Except for the two buttons of the last diagonal of the input wire, it is only aligned with the other uppermost buttons of this output wire, thus, the horizontal cuts need to be made, enforcing black, horizontal cuts in the rest of the wire. An analogous argument is used for the gray cut pattern. 
\fi

\ifabstract
The {\bf not gadget}, shown in \cref{fig:split-1p}(a), enables us to reverse the truth assignment in a variable wire.
\fi


\iffull
The {\bf not gadget}, shown in \cref{fig:split-1p}(a), enables us to reverse the truth assignment in a variable wire. (That is, the cut pattern switches from vertical/horizontal to diagonal and vice versa.)
Again, the last diagonal cut (gray) and the last horizontal cut (black) of the input wire from the left only feature two buttons in this wire. 
Moreover, there is again a single isolated button. If the gadget is entered with the black, vertical cut pattern, the isolated button cannot be cut with any cut from the input wire. The only remaining possibility is to cut the isolated button with the two leftmost diagonal buttons of the output wire, enforcing diagonal cuts in the rest of the wire. If the gadget is entered with the gray, diagonal cut pattern, the isolated button needs to be cut with the last diagonal wire cut (as otherwise its size of two would render it infeasible). In the output wire the first diagonal cut would be to short, enforcing horizontal cuts in the rest of the wire.
\fi

\iffull
\begin{figure}[t]
\center
    \comicII{.65\textwidth}{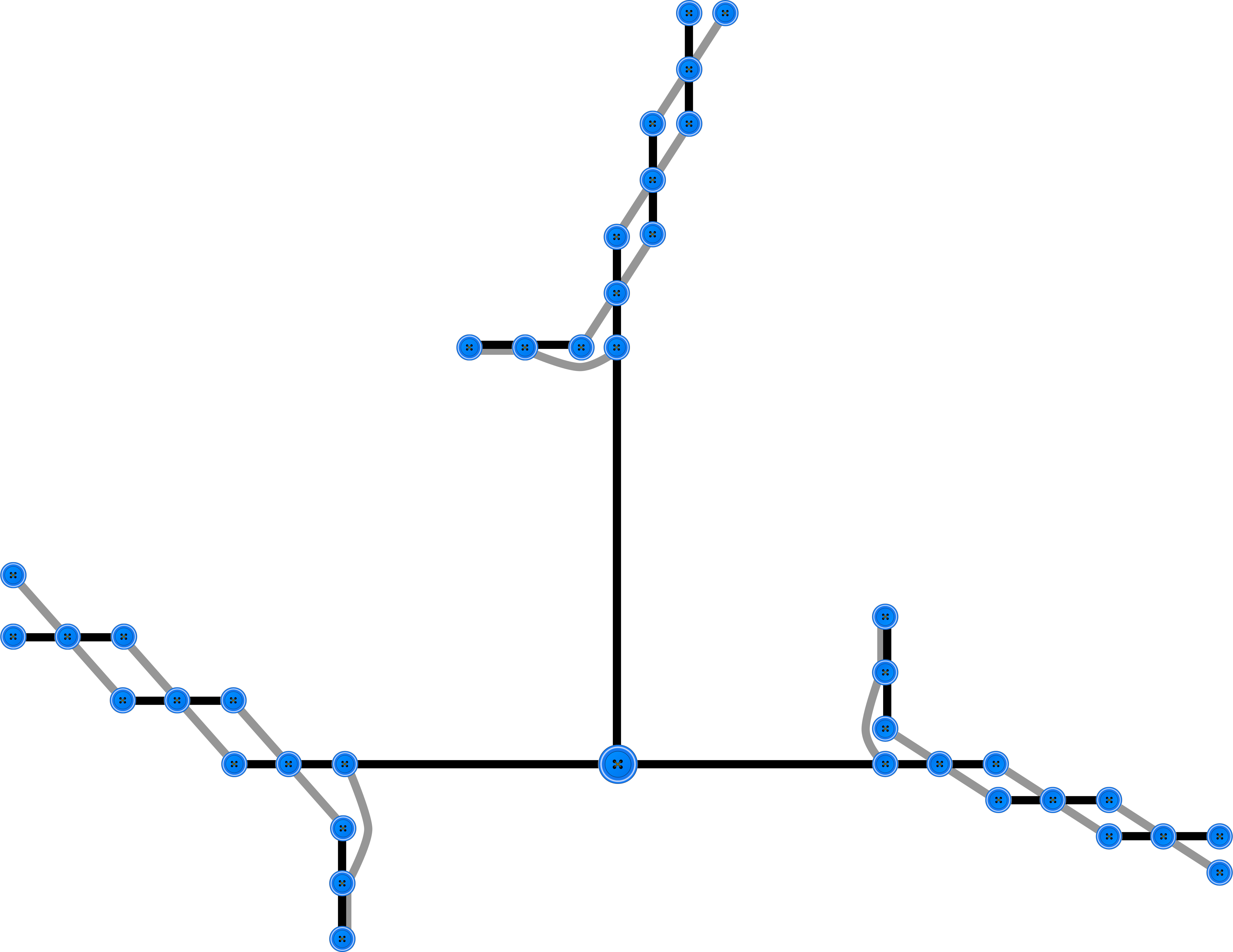} 
         \caption{\label{fig:clause-1p}\small The clause gadget for the 1-color case.
         }
\end{figure}
\fi

\ifabstract
The {\bf clause gadget} is shown in \cref{fig:var-1p}(c).
This gadget simulates a conjunction of literals.
\fi 

\iffull
The {\bf clause gadget} is shown in \cref{fig:clause-1p}. Wires from literals approach the gadget along three diagonals. The clause has one isolated button (shown enlarged), and two additional buttons for each wire that build a row of four aligned buttons with the last two buttons of the wire. When a gray cut direction is used, the two lower plus the topmost of this row of four combine for a feasible cut. If a black cut direction is used, the three lower buttons of this row allow a feasible cut. But this still leaves the isolated button to cut: Only if in at least one of the input wires the black cuts, corresponding to a truth setting of ``true'' for the literal, are used, this button can be added to the last black cut of the wire (extending its size from 3 to 4). If all literals are set to ``false'', i.e., along all three input wires the gray cut pattern is used, cutting the isolated button with two others along one of the horizontal/vertical axis, leaves at least one isolated button, rendering the complete board infeasible, as we are not left with an empty board. 
\fi

Thus, the resulting Buttons \& Scissors board has a solution if and only if at least one of the literals per clause is set to true, that is, if and only if the original PLANAR 3-SAT formula $F$ is satisfiable. It is easy to see that this reduction is possible in polynomial time. In addition, given a Buttons \& Scissors board and a sequence of cuts, it is easy to check whether those constitute a solution, i.e., whether all cuts are feasible and result in a board with only empty grid entries. Hence, $\Board{n \times n}{1}{\infty}{\dirsFour}{\{3,4\}}{B}$ is in the class NP. Consequently, $\Board{n \times n}{1}{\infty}{\dirsFour}{\{3,4\}}{B}$ is NP-complete.  \qed
\end{proof}

\subsubsection{Hardness for Cut Size $\{2\}$ and 2-Colors}

An intermediate problem is~below.

\noindent{\bf Decision Problem:} Graph Decycling on $(G, S)$. \\
{\bf Input:} Directed graph $G = (V, E)$ and a set of disjoint pairs of vertices $S \subseteq V \times V$.\\
{\bf Output:} True, if we can make $G$ acyclic by removing either $s$ or $s'$ from $G$ for every pair $(s, s') \in S$.  Otherwise, False.
  
\begin {lemma}\label{lem:decycling-2-color}
  Graph Decycling reduces to Buttons \& Scissors with $2$ colors.
\end {lemma}

\begin {proof}
  Consider an instance $(G, S)$ to graph decycling.
  First, we observe that we can assume that every vertex in $G$ has degree $2$ or $3$, and more specifically, in-degree $1$ or $2$, and out-degree $1$ or $2$.
  Indeed, we can safely remove any vertices with in- or out-degree $0$ without changing the outcome of the problem. Also, we can replace a node with out-degree $k$ by a binary tree of nodes with out-degree $2$. The same applies to nodes with in-degree $k$.
  
  Furthermore, we can assume that every vertex that appears in $S$ has degree~$2$.
  Indeed, we can replace any degree $3$ vertex by two vertices of degree $2$ and $3$, and use the degree $2$ vertex in $S$ without changing the outcome.
  Similarly, we can assume that no two vertices of degree $3$ are adjacent.
  Finally, we can assume that $G$ is bipartite, and furthermore, that all vertices that occur in $S$ are in the same half of $V$, since we can replace any edge by a path of two edges.
  
  \vierplaatjes [scale = .205] {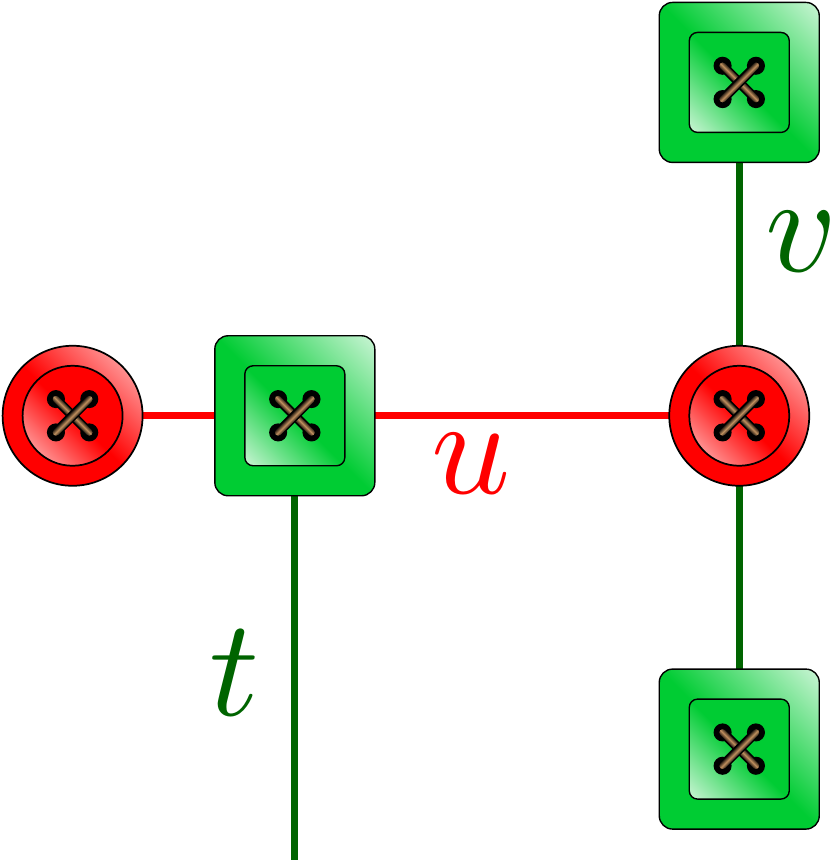} {2c_2-1} {2c_1-2} {2c_link} {Three types of nodes: (a) in-degree $1$ ($tu$) and out-degree $1$ ($uv$); (b) in-degree $2$ ($su$ and $tu$) and out-degree $1$ ($uv$); (c) in-degree $1$ ($tu$) and out-degree $2$ ($uv$ and $uw$).  In (d) the nodes $u$ and $v$ are linked in $S$ and we can choose to remove $u$ or $v$.}
  
  Now, we discuss how to model such a graph in a Buttons \& Scissors instance. Each node will correspond to a pair of buttons, either a red or a green pair according to a bipartition of $V$. These pairs of buttons will be mapped to locations in the plane on a common (horizontal for red, vertical for green) line, and such that any two buttons of the same color that are not a pair are not on a common (horizontal, vertical, or diagonal) line (unless otherwise specified). If two nodes of opposite colors $u$ and $v$ are connected by an edge in $G$, we say that $u$ \emph {blocks} $v$. In this case, one of the buttons of $u$ will be on the same line as the buttons of $v$, and more specifically, it will be between the two buttons of $v$. That is, $v$ can only be cut if $u$ is cut first. Buttons of opposite colors that are not connected by an edge will not be on any common lines either.
    
  As discussed above, we can assume we have only three possible types of nodes.
  \cref{fig:2c_1-1} illustrates the simplest case, of a node $u$ with one incoming edge $tu$ and one outgoing edge $uv$. Clearly, $t$ blocks $u$ and $u$ blocks $v$.
  To model a node with in-degree $2$, we need to put two buttons of different same-colored nodes on the same line (see \cref{fig:2c_2-1}). As long as the other endpoints of these two edges are not on a common line this is no problem: we never want to create a cut that removes one button of $s$ and one of $t$, since that would create an unsolvable instance.
  Finally, to model a node with out-degree $2$, we simply place a vertical edge on both ends of $u$ (see \cref{fig:2c_1-2}). Note that is it important here that we do not connect two nodes with out-degree $2$ to the same two nodes with in-degree $2$, since then we would have both pairs of endpoints on a common line; however, we assumed that nodes of degree $3$ are never adjacent so this does not occur.

  \ifabstract
  What is left is to create a mechanism to remove vertices from $G$ as dictated by $S$.  This mechanism is illustrated in \cref{fig:2c_link}, with full details given in Appendix B. \qed
  \fi

  \iffull
  We embed $G$ in this fashion. What is left is to create a mechanism to remove vertices from $G$ as dictated by $S$. We assume all vertices in $S$ correspond to green button pairs. Suppose we have edges $su$, $uw$, $tv$, and $vx$, and $u$ and $v$ are linked in $S$. We place six additional green buttons on two vertical lines as illustrated in \cref{fig:2c_link}.
  Three buttons are placed such that one is on the same horizontal line as the top end of $u$ and the other is on the same horizontal line as the top end of $v$, and the third is not on a common line with any other node. The third button assures we must use this vertical cut, but we can choose whether to remove the top of $u$ or the top of $v$. Similarly, we create a construction that allows us to remove the bottom of $u$ or the bottom of $v$. Since $u$ and $v$ are ordinary degree $2$ nodes, we never want to remove the top of one and the bottom of the other, since this would yield an unsolvable situation.
  Note that, in case we are allows to make cuts of size $3$, we can also remove all six extra nodes completely. This does not affect the reduction, since it never pays to do so. \qed
  \fi 
\end {proof}

\ifabstract
\later{
  \subsection*{Conclusion of \cref{lem:decycling-2-color}}
  We embed $G$ in this fashion.  We assume all vertices in $S$ correspond to green button pairs. Suppose we have edges $su$, $uw$, $tv$, and $vx$, and $u$ and $v$ are linked in $S$. We place six additional green buttons on two vertical lines as illustrated in \cref{fig:2c_link}.
  Three buttons are placed such that one is on the same horizontal line as the top end of $u$ and the other is on the same horizontal line as the top end of $v$, and the third is not on a common line with any other node. The third button assures we must use this vertical cut, but we can choose whether to remove the top of $u$ or the top of $v$. Similarly, we create a construction that allows us to remove the bottom of $u$ or the bottom of $v$. Since $u$ and $v$ are ordinary degree $2$ nodes, we never want to remove the top of one and the bottom of the other, since this would yield an unsolvable situation.
  Note that, in case we are allows to make cuts of size $3$, we can also remove all six extra nodes completely. This does not affect the reduction, since it never pays to do so.
}
\fi

\begin{lemma}\label{lem:decycling-hard}
  SAT reduces to Graph Decycling.
\end{lemma}

\ifabstract
A proof of \cref{lem:decycling-hard} is in Appendix B.
\later {
\subsection*{Proof of \cref{lem:decycling-hard}}
\fi
\begin{proof}
Create a cycle for each clause, and mark a vertex in it for each literal in the clause: one of these needs to be removed to make the graph acyclic.
Create a pair of removable vertices for each variable. 
Duplicate literals with little bow-tie gadgets: two cycles that share a vertex. 
If shared vertex is \emph{not} removed, we must remove two more vertices, one from each cycle. \qed
\end{proof}
\ifabstract
} 
\fi

Lemmas \ref{lem:decycling-2-color} and \ref{lem:decycling-hard} give \cref{thm:2colors2cuts}.

\begin{theorem} \label{thm:2colors2cuts}
\ifabstract
$\Board{n \times n}{2}{\infty}{\dirsTwoA}{\{2\}}{B}$ is NP-complete.
\else
Buttons $\&$ Scissors for 2-colors with only cuts of two directions of size 2, i.e., $\Board{n \times n}{2}{\infty}{\dirsTwoA}{\{2\}}{B}$ is NP-complete.
\fi
\end{theorem}




\section{Two-Player Games}\label{sec:2p}


\ifabstract
We consider three two-player Buttons $\&$ Scissors variants. 
First we consider color restricted games where (a) each player can only cut specific colors, and (b) players are not restricted to specific colors. For (a) player blue may only cut Blue buttons, while the red player may only cut Red buttons. For (b) we distinguish by winning criterion: for (\impartial) the last player who makes a feasible cut wins; for (\scoring) players keep track of the total number of buttons they've cut.  When no cuts can be made, the player with the most buttons cut wins.

In the following sections, we show that all variants are \cclass{PSPACE}-complete.
\fi

\iffull
We consider four different types of two-player Buttons $\&$ Scissors variants, in two different types:

\begin{itemize}
    \item Games where each player may only make cuts on specific colors:
	    \begin{enumerate}[nolistsep,noitemsep]
	         \item One player may only cut Blue buttons, while the other player may only cut Red buttons. 
	         \item Either player may cut Green buttons.
	    \end{enumerate}
	\item Games where players are not restricted to the colors they may cut:
	    \begin{enumerate}[nolistsep,noitemsep]
	        \item The last player who makes a feasible cut wins.  (\impartial)
	        \item Players keep track of the total number of buttons they've cut.  After no more cuts can be made, the player with the highest number of buttons cut wins.  (\scoring) 
	    \end{enumerate}
\end{itemize}

In the following sections, we show that all four variants are \cclass{PSPACE}-complete.
\fi

\subsection{Cut-By-Color Games}\label{sec:2p_partisan}

In this section the first player can only cut blue buttons, the second player can only cut red buttons, and the last player to make a cut wins.

\begin{theorem}\label{th:2p-part}
The partisan LAST two-player Buttons $\&$ Scissors game, where one player cuts blue buttons, the other red buttons, is PSPACE-complete. 
\end{theorem}
\begin{proof}
The proof is by reduction from $G_{\%\mathit{free}}(\text{CNF})$~\cite{schaefer}: given a boolean formula $\Phi(x_1,\dots,x_n)$ in CNF and a partition of the variables into two disjoint subsets of equal size $V_b$ and $V_r$, two players take turns in setting values of the variables, the first (Blue) player sets the values of variables in $V_b$, and the second (Red) player sets the values of variables in $V_r$. Blue wins if, after all variables have been assigned some values, formula $\Phi$ is satisfied, and loses otherwise.


For a given instance of formula $\Phi$ we construct a Buttons $\&$ Scissors board $B$, such that Blue can win the game on $B$ if and only if he can satisfy formula $\Phi$. We will prove this statement in different formulation: Red wins the game on $B$ if and only if formula $\Phi$ cannot be satisfied.
See \cref{fig:2pex} for a full example. 

\iffull
\begin{figure}[t]
\center
     \comic{.15\textwidth}{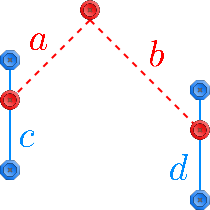}{(a)} 
     \hfill
    \comic{.2\textwidth}{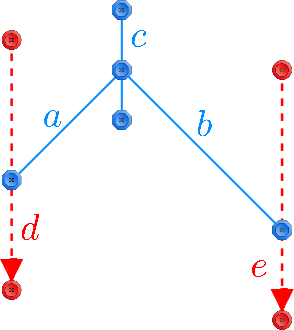}{(b)} 
    \hfill
     \comic{.12\textwidth}{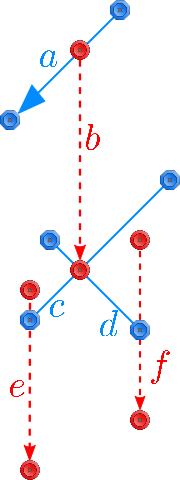}{(c)}\\ 
      \vspace*{-1cm}
      \comic{.25\textwidth}{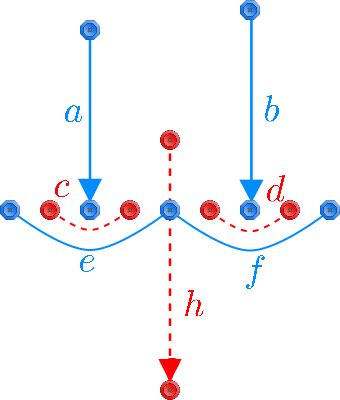}{(d)} 
    \hfil
      \comic{.22\textwidth}{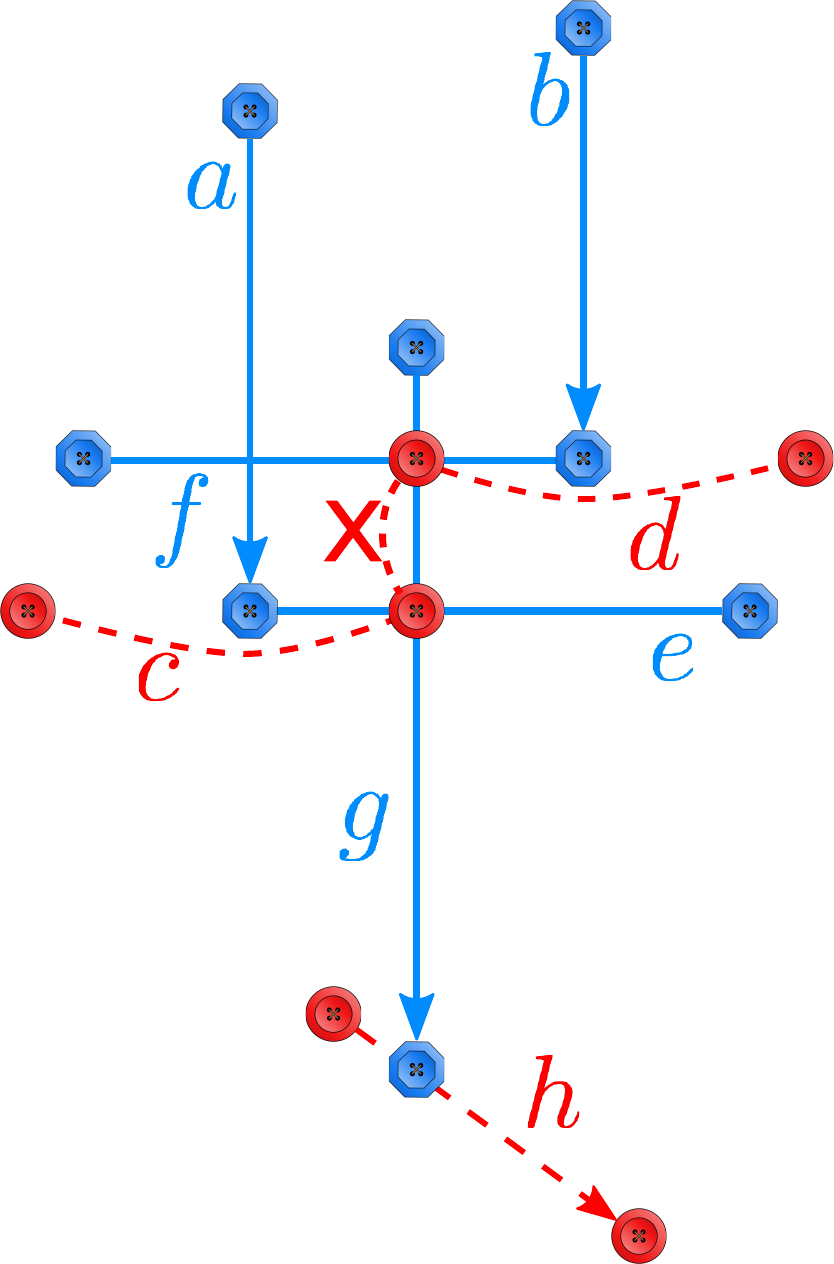}{(e)} 
         \caption{\label{fig:2p-part}\small (a) The red (dashed) variable gadget, (b) the blue (solid) variable gadget, (c) the split gadget, (d) the OR gadget, and (e) the AND gadget. Lines (or arcs used for clarity) indicate which buttons are aligned.
         }
\end{figure}
\fi

\ifabstract
\vijfplaatjes [scale = .26] {2p_part_varred} {2p_part_varblue} {2p_part_split} {2p_part_or} {2p_part_and} {\small (a) The red (dashed) variable gadget, (b) the blue (solid) variable gadget, (c) the split gadget, (d) the OR gadget, and (e) the AND gadget. Lines (or arcs used for clarity) indicate which buttons are aligned.}
\fi

\begin{figure}
\center
    \comicII{.78\textwidth}{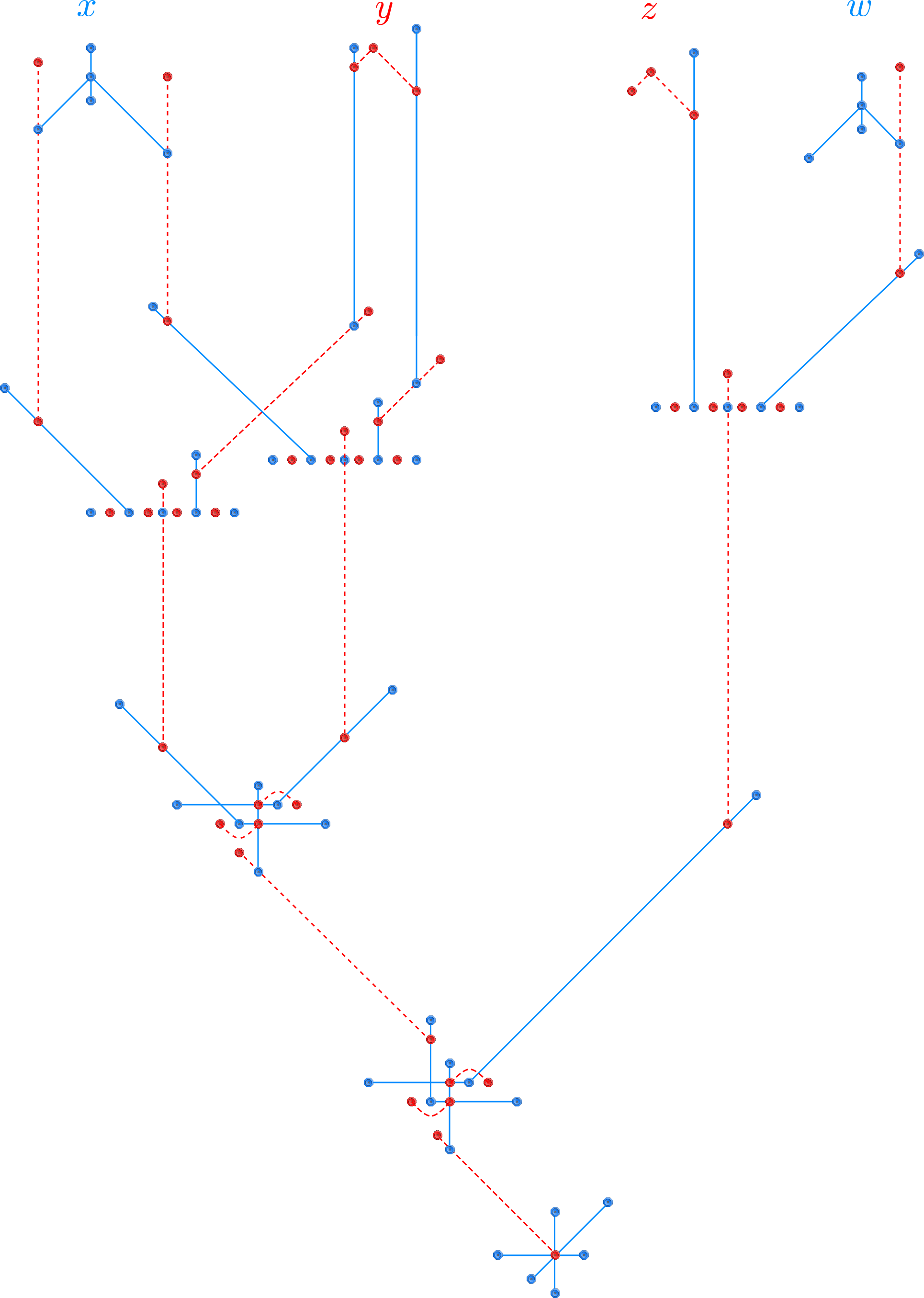} 
         \caption{\label{fig:2pex}\small Example for the construction from the proof of \cref{th:2p-part}: we consider \newline$({\color{blue}x}\vee {\color{red}y}) \wedge ({\color{blue} \neg x} \vee {\color{red}\neg y})\wedge({\color{red}z}\vee {\color{blue}w})$.
         }
\end{figure}


The {\bf red variable gadget} is shown in \cref{fig:2p-part} (a). Red ``sets the value'' of the corresponding variable by choosing the first cut to be $a$ (false) or $b$ (true), and thus unlocking one of the two cuts, $c$ or $d$, respectively, for Blue to follow up (and to propagate the value of the variable).

The {\bf blue variable gadget} is shown in \cref{fig:2p-part} (b). Blue ``sets the value'' of the corresponding variable by choosing the first cut to be $a$ (false) or $b$ (true), and thus unlocking one of the two cuts, $d$ or $e$, respectively, for the red player to follow up. Blue has one extra cut $c$ that is used to pass the turn to Red. Alternatively, Blue can choose to start with the 3-button cut $c$ and disallow Red from making any cuts in the gadget. In that case the corresponding variable cannot be used to satisfy $\Phi$.



\cref{fig:2p-part} (d) depicts the {\bf OR} gadget: if Blue cuts $a$ or $b$ (or both), Red can leave the gadget with cut $h$. Cuts $a$ and $b$ unblock cuts $c$ and $d$, respectively, which in turn unblock $e$ and $f$, respectively. 

\cref{fig:2p-part} (e) depicts the {\bf AND} gadget for two inputs. The proper way of passing the gadget: Blue makes both cuts $a$ and $b$, and Red makes cuts $c$ and $d$ when they get unblocked, thus enabling Blue to make cut $g$ and exit the gadget. However, Red could also take an ``illegal'' cut $x$, thus, unblocking two extra cuts, $e$ and $f$, for the blue player, and, hence, putting Red at a disadvantage. Thus, if at any point in the game Red chooses (or is forced to) make cut $x$ in any of the AND gadgets, the 
game result is predetermined, and Red cannot win on $B$.

\cref{fig:2p-part} (c) shows the {\bf split} gadget; it enables us to increase the number of cuts leaving a variable and propagating its truth assignment. Blue's cut $a$ unblocks Red's cut $b$, which unblocks both $c$ and $d$. If Blue cuts $c$ and $d$ this enables Red to cut $e$ and $f$, respectively. The gadget also exists with Blue and Red reversed.

A variant of the split gadget evaluates the formula $\Phi$: cuts $e$ and $f$ are deleted. If the  variable values are propagated to this gadget and Red is forced to make the cut $b$, Blue then gets extra cuts which Red will not be able to follow up. 

The game progresses as follows: Blue selects an assignment to a blue variable. This unlocks a path of red-blue cuts that goes through some AND and OR gadgets and leads to the final gadget. As the order of the cuts in such a path is deterministic, and does not affect the choice of values of other variables, w.l.o.g., we assume that Red and Blue make all the cuts in this path (until it gets ``stuck'') before setting the next variable. The path gets stuck when it reaches some AND gadget for which the other input has not been cleared. The last cut in such a path was made by Red, thus afterwards it will be Blue's turn, and he may choose to make the leftover cut $c$ from the variable gadget to pass the turn to Red.

If the final gadget is not unblocked yet, Red always has a cut to make after Blue makes a move, as there is the same number of blue and red variables. However, if Blue can force Red to make moves until the final gadget is reached, then Blue gets extra cuts; thus, Red will run out of moves and lose the game. Otherwise, if Blue cannot fulfill some AND or OR gadgets, the Red player will make the last move and win. Therefore, if $\Phi$ cannot be satisfied, Red wins. \qed
\end{proof}

\subsection{Any Color Games}\label{sec:2p_impartial}

\iffull
\begin{theorem}
  \label{thm:impartialLastHard}
Impartial two-player Buttons $\&$ Scissors (\impartial) is PSPACE-complete. 
\end{theorem}

\begin{theorem}
  \label{thm:scoringHard}
Scoring two-player Buttons $\&$ Scissors (\scoring) is PSPACE-complete. 
\end{theorem}
\fi

\ifabstract
\begin{theorem}
  \label{thm:impartialLastHard}
\impartial\ two-player Buttons $\&$ Scissors is PSPACE-complete. 
\end{theorem}

\begin{theorem}
  \label{thm:scoringHard}
\scoring\ two-player Buttons $\&$ Scissors is PSPACE-complete. 
\end{theorem}
\fi

We show that \impartial\ is \cclass{PSPACE}-complete, then use one more gadget to show \scoring\ is \cclass{PSPACE}-complete.  We reduce from \ruleset{Geography}\footnote{Specifically, \ruleset{Directed Vertex Geography}, usually called \ruleset{Geography}.}, (\cclass{PSPACE}-complete \cite{LichtensteinSipser:1980}).  We use \cref{lemma:lowDegreeGeographyHard} to start with low-degree \ruleset{Geography} instances.

\begin{lemma}
  \label{lemma:lowDegreeGeographyHard}
  \ruleset{Geography} is \cclass{PSPACE}-complete even when vertices have max degree 3 and the max in-degree and out-degree of each vertex is 2.
\end{lemma}

\ifabstract
\cref{lemma:lowDegreeGeographyHard} is proven in Appendix B.  
\later{
\subsection*{Proof of \cref{lemma:lowDegreeGeographyHard}}
\fi
\begin{proof}
The proof is simply a reduction that adds some vertices to reduce the degree.  To reduce the in-degree, we use the gadget shown in \cref{fig:geographyReduceInDegree} to iteratively reduce the number of edges entering any vertex, $v$, with in-degree above 3.  Notice that the player who moves to $v$ will be the same in both games.

\begin{figure}[h!]
  \begin{center}
    \includegraphics[scale = .3]{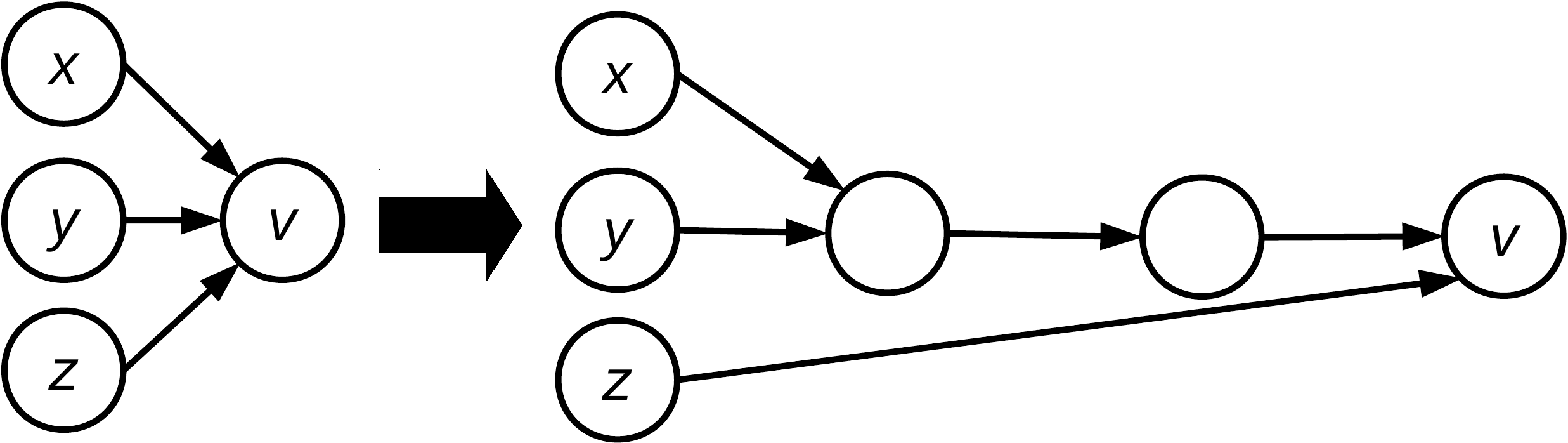} 
  \end{center}
  \caption{Reducing the in-degree to $v$ by 1 in \ruleset{Geography}.}
  \label{fig:geographyReduceInDegree}
\end{figure}

To reduce the out-degree, use use an analagous construction, as shown in \cref{fig:geographyReduceOutDegree}.  Just as in the previous gadget, we ensure that the same player will arrive at $x$, $y$, and $z$ as in the original game.  Additionally, we ensure that the player that leaves $v$ still makes the choice between $x$ and $y$ in the reduced gadget.

\begin{figure}[h!]
  \begin{center}
    \includegraphics[scale = .3]{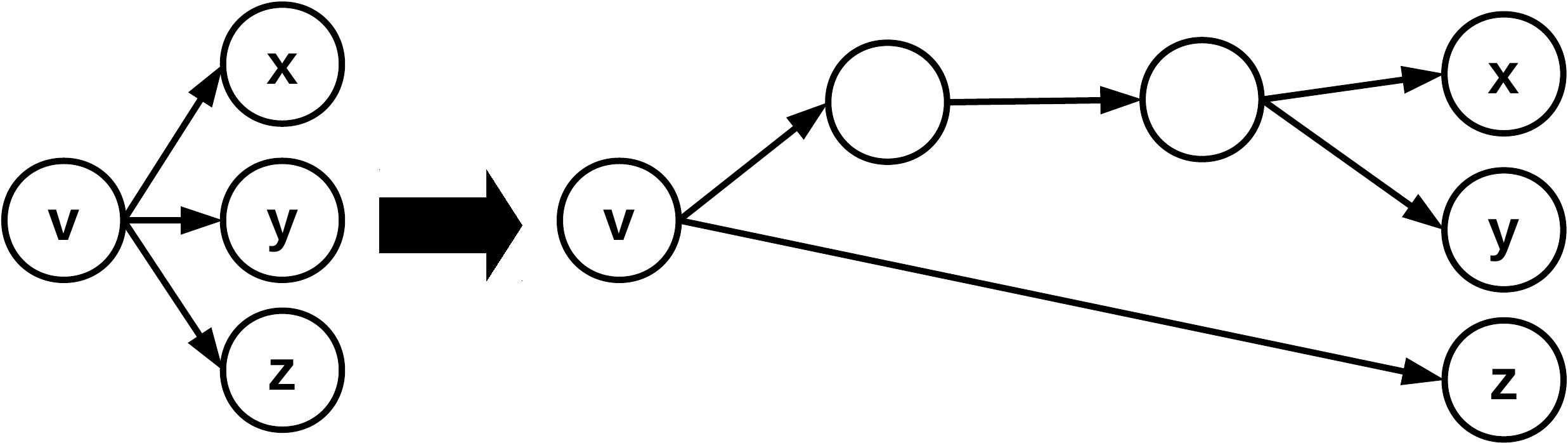} 
  \end{center}
  \caption{Reducing the out-degree of $v$ by 1 in \ruleset{Geography}.}
  \label{fig:geographyReduceOutDegree}
\end{figure}

After repeatedly reducing the in and out degree to at most 2, some vertices may still have both in and out degree 2.  To fix this, we simply add two vertices as shown in \cref{fig:geography2And2}. \qed

\begin{figure}[h!]
  \begin{center}
    \includegraphics[scale = .3]{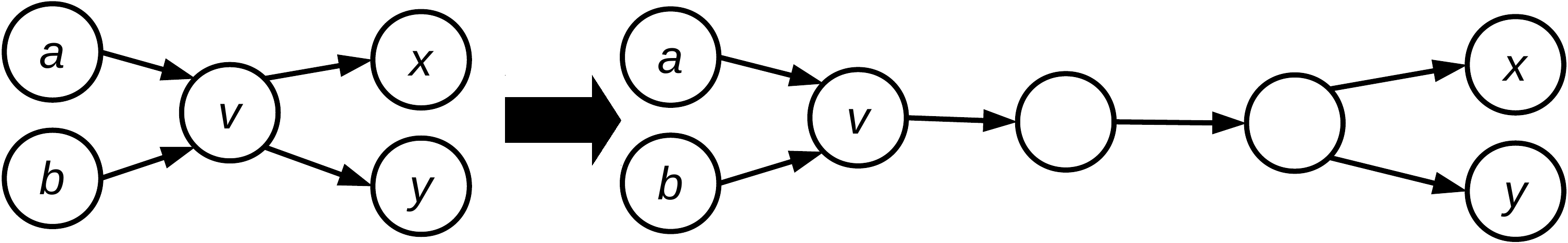} 
  \end{center}
  \caption{Final step for vertices, $v$, with both in and out-degree of 2.}
  \label{fig:geography2And2}
\end{figure}
\end{proof}
\ifabstract
} 
\fi
The following gadgets prove \cref{thm:impartialLastHard}.

\ifabstract
\begin{figure}[t]
\centering
\comic{.42\textwidth}{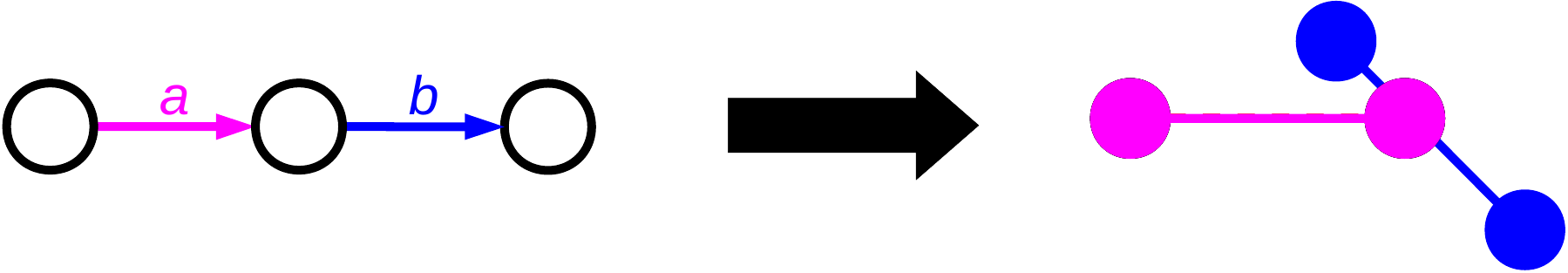}{(a)} 
\hfill
\comic{.42\textwidth}{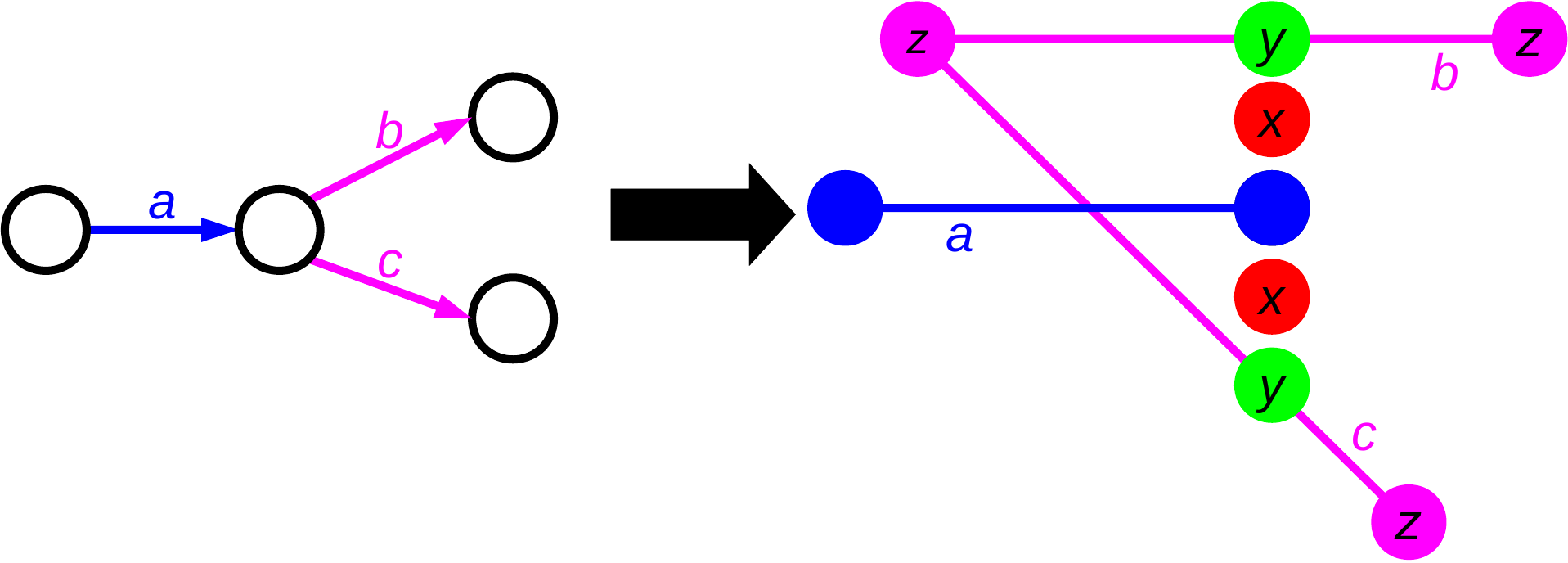}{(b)}\\ 
\comic{.42\textwidth}{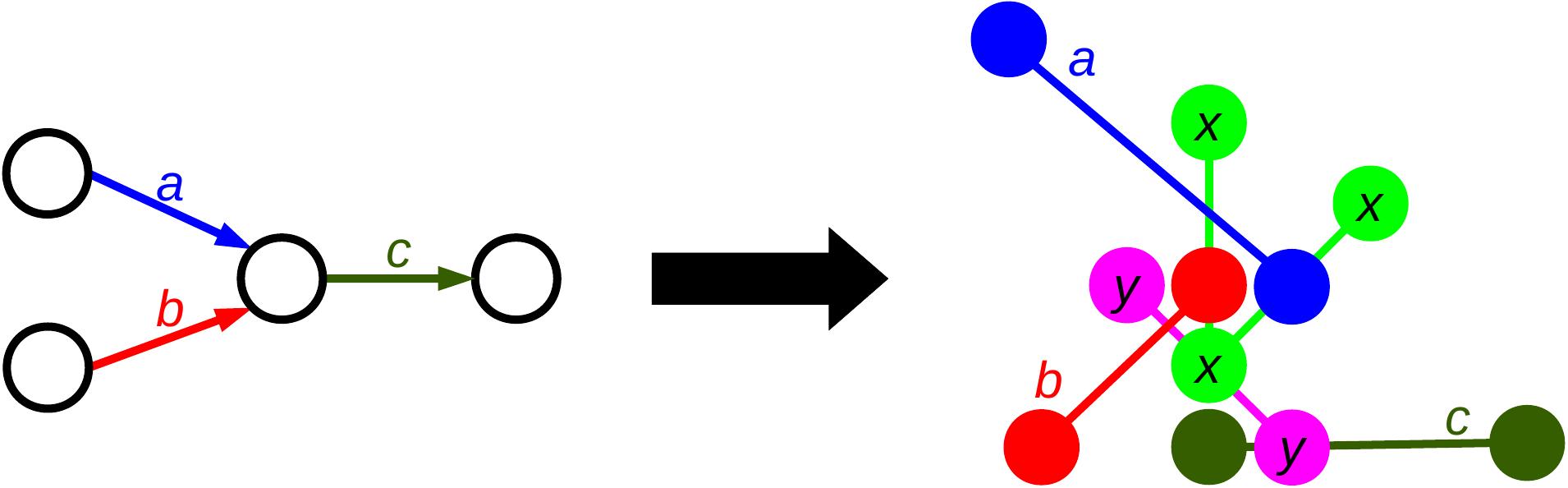}{(c)} 
\hfill
\vspace*{-.2cm}
\comic{.30\textwidth}{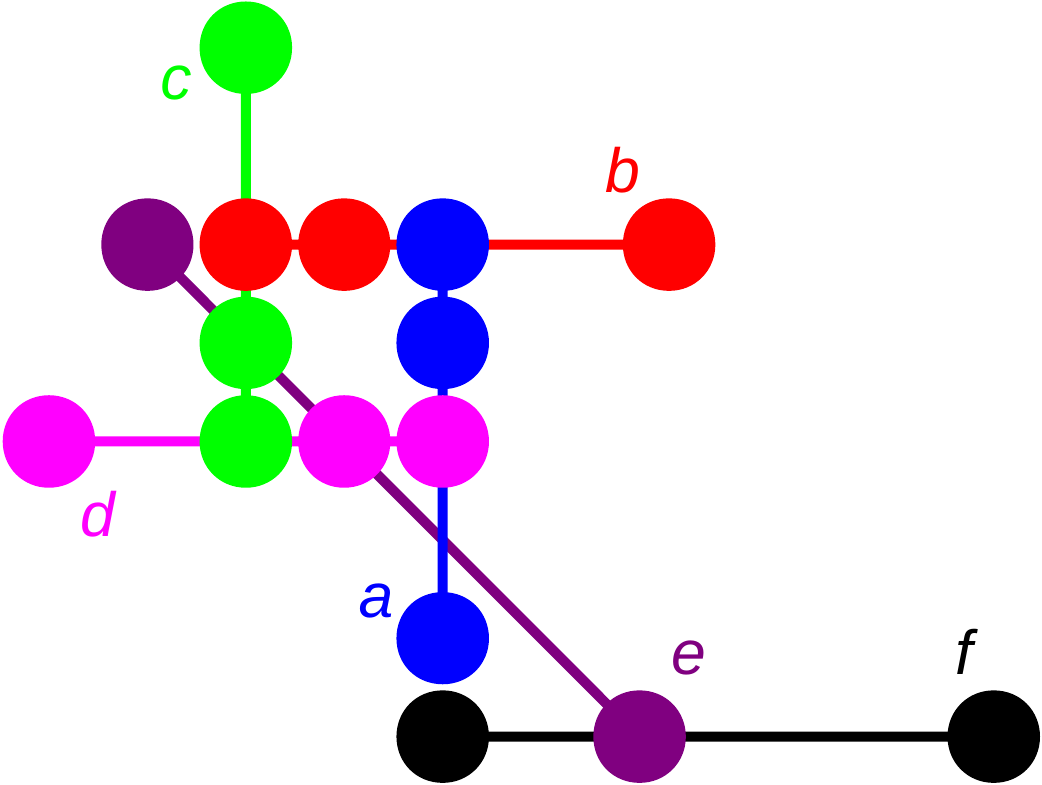}{(d)}\hspace*{.7cm} 
\caption{Reduction gadgets for vertex with (a) one incoming arc and one outgoing arc, (b) one incoming arc and two outgoing arcs, and (c) two incoming arcs and one outgoing arcs. (d) The starting gadget for \scoring.}
\label{fig:twoPlayerImpartial}
\end{figure}
\fi
\iffull
\begin{figure}[t]
\centering
\comic{.47\textwidth}{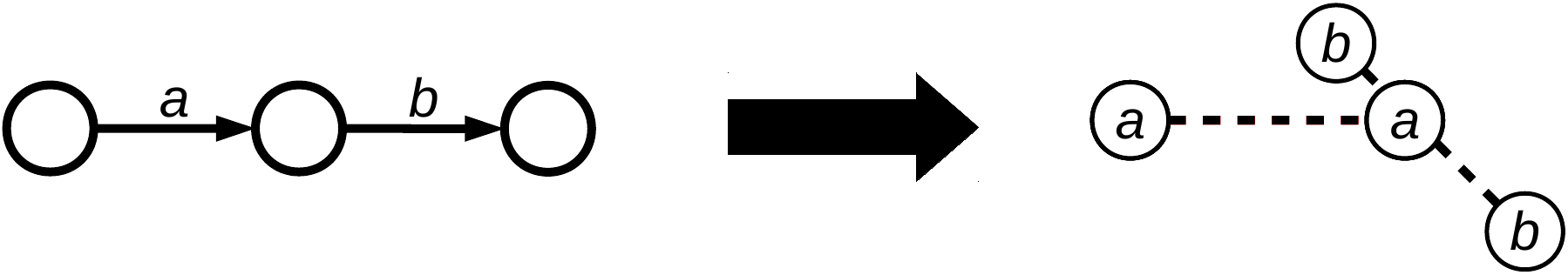}{(a)} 
\hfill
\comic{.47\textwidth}{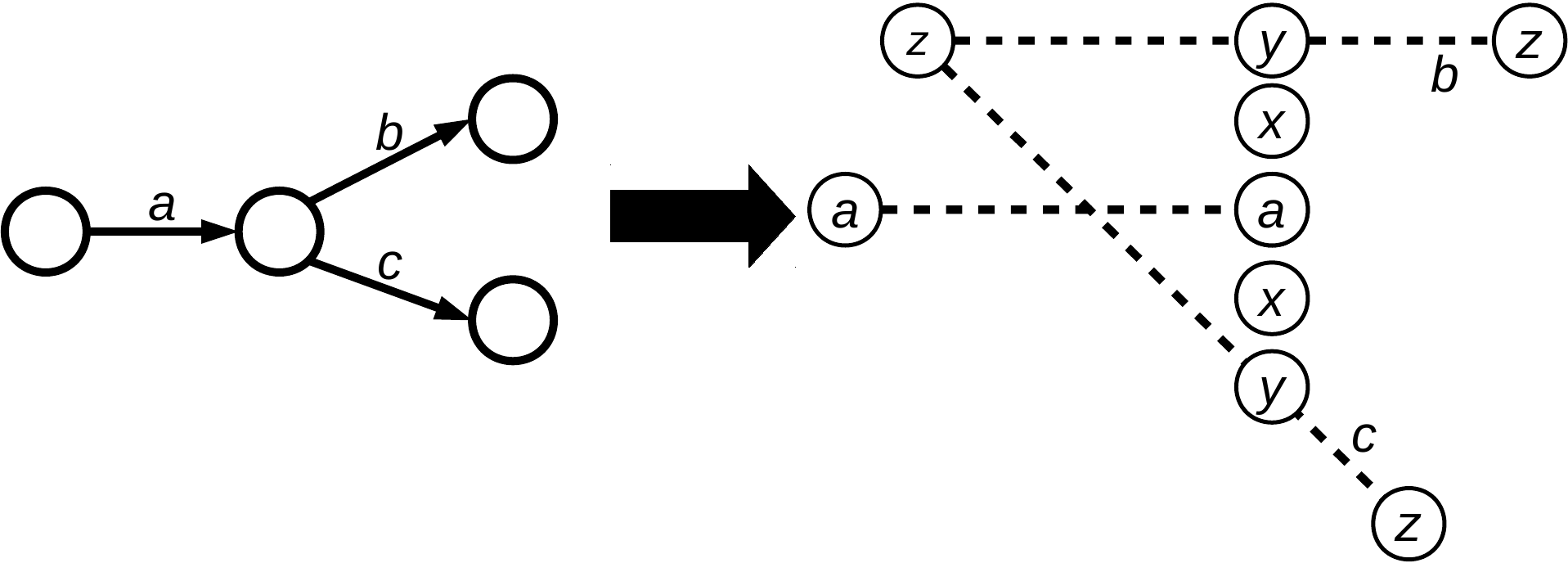}{(b)} 
\comic{.47\textwidth}{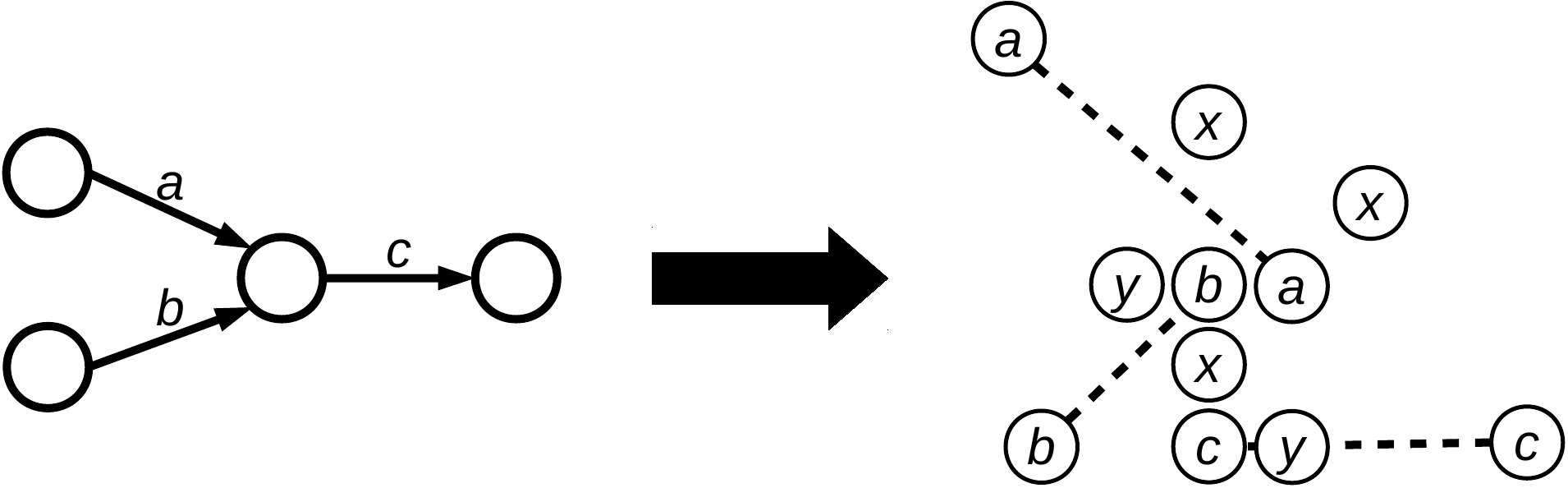}{(c)} 
\caption{Reduction gadgets for vertex with (a) one incoming arc and one outgoing arc, (b) one incoming arc and two outgoing arcs, and (c) two incoming arcs and one outgoing arcs.}
\label{fig:twoPlayerImpartial}
\end{figure}
\fi

\ifabstract
\begin{itemize}
  \item \textbf{In-degree 1, out-degree 1}: The gadget for this is a pair of buttons such that removing the first pair frees up the second, as in \cref{fig:twoPlayerImpartial}(a).
  
  \item \textbf{In-degree 1, out-degree 2}: See \cref{fig:twoPlayerImpartial}(b) and Appendix B.
  
  \item \textbf{In-degree 2, out-degree 1}: See \cref{fig:twoPlayerImpartial}(c) and Appendix B.  
  
  \item \textbf{In-degree 0}: The gadgets for this look just like the gadgets for the analagous in-degree 1 gadgets, but without the button pair for the incoming edge.
  
  \item \textbf{Out-degree 0}: Each edge is a button pair that won't free up other buttons.
\end{itemize}

\later{
\subsection*{Discussion of Gadgets for Proof of \cref{thm:impartialLastHard}}
\else
\begin{itemize}
  \item \textbf{In-degree 1, out-degree 1}: The gadget for this is a pair of buttons such that removing the first pair frees up the second, as in \cref{fig:twoPlayerImpartial}(a).
  \item \textbf{In-degree 1, out-degree 2}: The gadget for this is shown in \cref{fig:twoPlayerImpartial}(b).  Incoming edge $a$ is represented by the button pair colored $a$.  The outgoing edges $b$ and $c$ are represented by three non-colinear buttons colored $z$.  After the first player removes the button pair $a$, the second player will cut the pair colored $x$, and the first player will cut the $y$ pair.  Just as the second player would choose between moving to $b$ or $c$ in the \ruleset{Geography} position, so do they choose between cutting the $z$ pair labelled $b$ or the $z$ pair labelled $c$.  Once one of those pairs is chosen, the other $z$-colored button is stranded.
  \item \textbf{In-degree 2, out-degree 1}: The gadget for this is shown in \cref{fig:twoPlayerImpartial}(c).  Incoming \ruleset{Geography} arcs $a$ and $b$ are represented by button pairs $a$ and $b$, respectively.  No matter which of those two are cut, it allows the $x$ pair to be cut, followed by the $y$ pair, then the $c$ pair, representing the $c$ arc in the \ruleset{Geography} position. The player that cuts the $c$ pair will be the opposite of the player that cuts the $a$ or $b$ pair, just as in the \ruleset{Geography} position.
  \item \textbf{In-degree 0}: The gadgets for this look just like the gadgets for the analagous in-degree 1 gadgets, but with the button pair missing that represents the incoming edge.
  \item \textbf{Out-degree 0}: Each edge is a button pair that doesn't free up other buttons upon removal.
\end{itemize}
\fi
\ifabstract
} 
\fi

\iffull
\begin{proof}
  (For \cref{thm:impartialLastHard}.)  
  We reduce from our restricted-degree version of \ruleset{Geography}.  The reduction replaces each \ruleset{Geography} vertex, $v$, with the cooresponding gadget described above.  Each arc between two vertices becomes a button pair (as described in the reduction) used in the gadgets for both vertices. \qed
  
\end{proof}
\fi
\ifabstract
\later{
\subsection*{Proof of \cref{thm:impartialLastHard}}
\begin{proof}
We reduce from our restricted-degree version of \ruleset{Geography}.  The reduction replaces each \ruleset{Geography} vertex, $v$, with the cooresponding gadget described above.  Each arc between two vertices becomes a button pair (as described in the reduction) used in the gadgets for both vertices. \qed

\end{proof}
}
We prove \cref{thm:impartialLastHard} in Appendix B.
\fi
To show \scoring\ is hard, we create a reduction where after each turn, that player will have cut the most buttons; the last player to move wins.  This alternating-advantage situation is caused by an initial gadget.  The optimal play sequence begins by cutting two buttons, then three, then three, then three a final time.  After these first four moves, the first player will have five points and the second player six.  Each subsequent cut removes two buttons so each turn ends with the current player ahead.

\iffull
\begin{figure}[h!]
  \begin{center}
	\includegraphics[scale = .4]{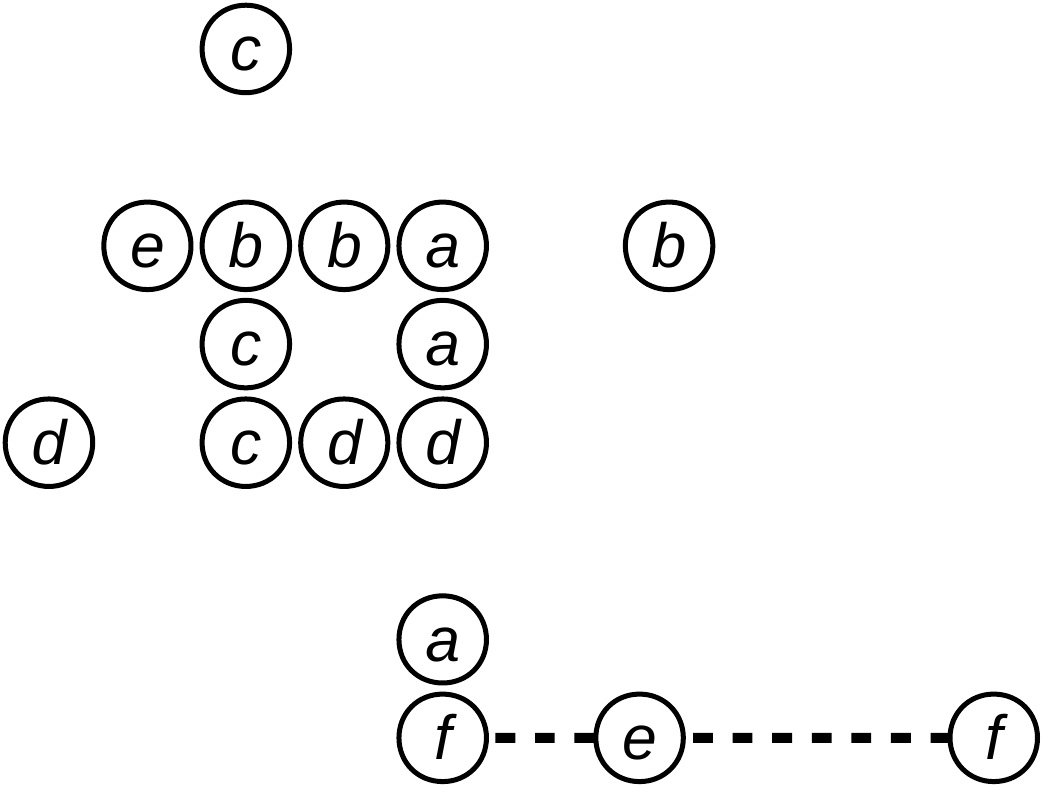} 
  \end{center}
  \caption{The starting gadget for \scoring.}
  \label{fig:scoringStart}
\end{figure}
\fi

\iffull
\cref{fig:scoringStart} shows the starting gadget that sets up this initial back-and-forth.  The color-$f$ buttons will be the last two cut; the right-hand $f$ button must be blocking the next gadget. \cref{lemma:scoringStartMoves} postulates that $f$ will be last.
\fi

\ifabstract
\cref{fig:twoPlayerImpartial}(d) shows the starting gadget that sets up this initial back-and-forth.  The color-$f$ buttons will be the last two cut; the right-hand $f$ button must be blocking the next gadget. \cref{lemma:scoringStartMoves} postulates that $f$ will be last.
\fi

\begin{lemma}
  \label{lemma:scoringStartMoves}
  If a player has a winning strategy, then part of that winning strategy includes cutting all possible buttons of colors $a$, $b$, $c$, $d$, and $e$ before cutting $f$.  
\end{lemma}

\ifabstract
We provide a proof of \cref{lemma:scoringStartMoves} in Appendix B.

\later {
\subsection*{Proof of \cref{lemma:scoringStartMoves}} 
\fi
\begin{proof}
  The first player only earns two points by cutting any of $a$, $b$, $c$, or $d$.  After this first move, there is always another three-point move to take until all of $a$, $b$, $c$, and $d$ are gone.  Thus there are three chances to earn three points.  This gadget is the only place for a player to earn three points in one turn, so the second player must take two of those three triple button cuts.  If they don't, they will never have more points than the first player.  The second player must spend each of their first two moves cutting three buttons.
  
  Similarly, the first player cannot spend the first three turns earning only two points each time.  If they do, then the second player can take three on each of their first three turns, making the score 6 to 9 in favor of the second player.  Since all moves thereafter only earn 2 points, the second player will always have more points.  The first player must cut three buttons on their second or third turn.  
  
  Thus, after the first five turns (first-second-first-second-first) the score must be 7 to 6 in favor of the first player.  The only way for this to happen is if $f$ is cut after $a$, $b$, $c$, $d$, and $e$.  \qed
\end{proof}
\ifabstract
} 
\fi

With the lemmas in place, we can prove \cref{thm:scoringHard}.
\ifabstract
We provide the proof of \cref{thm:scoringHard} in Appendix B.

\later{
\subsection*{Proof of \cref{thm:scoringHard}}
\begin{proof}
  We reduce again from the low-degree \ruleset{Geography} game considered in \cref{lemma:lowDegreeGeographyHard}.  Use the reduction to \impartial\ described earlier to generate a Buttons $\&$ Scissors board $B$.  Add the starting gadget for \scoring\ (shown in \iffull\cref{fig:scoringStart}\fi \ifabstract\cref{fig:twoPlayerImpartial}(d)\fi) to create $B'$ so that the buttons marked 'f' 
  block the first move that can be made on $B$.
  
  After the buttons 'f' are cut, the second player will have exactly one more point than the first.  Throughout all the gadgets used in the \cclass{PSPACE}-hard reduction for \impartial, each cut removes two buttons.  Thus the last player who can win playing \impartial\ on $B$ can win playing \scoring\ on $B'$. \qed 
\end{proof}
}
\fi

\iffull
\begin{proof}
  (Proof of \cref{thm:scoringHard}.)  We reduce again from the low-degree \ruleset{Geography} game considered in \cref{lemma:lowDegreeGeographyHard}.  Use the reduction to \impartial\ described earlier to generate a Buttons $\&$ Scissors board $B$.  Add the starting gadget for \scoring\ (shown in \iffull\cref{fig:scoringStart}\fi \ifabstract\cref{fig:twoPlayerImpartial}(d)\fi) to create $B'$ so that the buttons marked 'f' 
  block the first move that can be made on $B$.
  
  After the buttons 'f' are cut, the second player will have exactly one more point than the first.  Throughout all the gadgets used in the \cclass{PSPACE}-hard reduction for \impartial, each cut removes two buttons.  Thus the last player who can win playing \impartial\ on $B$ can win playing \scoring\ on $B'$. \qed 
\end{proof}
\fi

\section{Open Problems}\label{sec:open}

Interesting problems for boards with a constant number of rows are still open.
A conjecture for $m=2$ rows appears below.

\begin{conjecture} \label{con:2rows}
There is a polynomial time algorithm that removes all but $s$ buttons from any full $2\times n$ board with $C=2$ colors for some constant $s$.
\end{conjecture}


\newpage

\ifabstract
\section*{Appendix A}

\section*{Appendix B}
\fi

\magicappendix

\end{document}